\documentclass{amsart}

\usepackage{pgf,tikz,xypic,graphicx,mathrsfs,etex,bm,sidecap,amssymb,hyperref,amsmath}
\usepackage{slashed}
\usepackage{tikz-cd}
\usepackage{stmaryrd} %For \varowedge
\xyoption{all}
\usetikzlibrary{snakes, shapes,arrows,patterns,decorations.pathreplacing,matrix}

\newtheorem{theorem}{Theorem}

\newtheorem{remark}{Remark}
\newtheorem{proposition}{Proposition}
\newtheorem{lemma}{Lemma}
\newtheorem{corollary}{Corollary}
\newtheorem{definition}{Definition}

\newcommand{\op}{\operatorname}

\newcommand{\monib}[2]{\binom{#2}{#1}}
\newcommand{\sgn}{\operatorname{sgn}}
\newcommand{\ind}{\operatorname{ind}}
\newcommand{\idx}{\operatorname{idx}}

\title[Sachs equations and plane waves, V]{Sachs Equations and Plane Waves, V: Ward, Fourier, and Heisenberg Symmetry on Plane Waves}
\date{\today}
\author{Jonathan Holland}
\author{George Sparling}
\address{University of Pittsburgh\\
  Department of Mathematics\\
  301 Thackeray Hall\\ Pittsburgh, PA 15260
}

\begin{document}

\begin{abstract}
This article studies wave equations and their solutions on plane wave spacetimes of arbitrary
dimension, developing the interplay among three structural layers: the Ward progressing-wave
representation of solutions to the scalar wave equation, the Fourier analysis of the Heisenberg
group naturally associated to the plane wave, and the Schr\"odinger propagator governing
the evolution of initial data.  The central geometric object is a positive curve in the Lagrangian
Grassmannian determined by the plane wave metric, previously studied in the authors' series
\cite{HS1,HS2,HS3,HS4}.  The conformal tensor $H(u)$ that parametrises this curve plays a
dual role: it encodes the null-cone geometry of the spacetime and simultaneously appears as
the time-dependent parameter in the Schr\"odinger representation of the Heisenberg group
acting by isometries on the plane wave.  Parallel to the classical Fourier inversion theorem,
convolution by Lagrangian delta distributions on the Heisenberg group furnishes an intrinsic
description of the Schr\"odinger propagator, and the intertwining of different polarisations by
this propagator is captured by a diagram that commutes up to a Maslov phase.  The theta
functions and Bargmann transforms that arise from imaginary polarisations complete the
analytic picture, connecting the present work to the theory of the Weil representation as
developed by Lion--Vergne \cite{LionVergne} and to Mumford's systematic treatment of theta
functions \cite{MumfordI,MumfordIII}.
\end{abstract}

\maketitle

\section{Introduction}

\subsection{Background and motivation}

Plane wave spacetimes occupy a distinguished position in both mathematics and physics.
Geometrically, they are the simplest genuinely curved Lorentzian manifolds: in the
Brinkmann--Rosen coordinate system the metric takes the form
$\mathcal R(G) = 2\,du\,dv - dx^T G(u)\,dx$,
where $G(u)$ is a smooth curve of positive-definite symmetric matrices.
The entire curvature of the spacetime is encoded by a single matrix-valued profile $K(u)$, making these spacetimes an ideal laboratory for explicit, rigorous analysis.

From the physics perspective, plane waves have been studied since the work of
Brinkmann \cite{Brinkmann} and Einstein--Rosen \cite{EinsteinRosen}, and they arise
as the \emph{Penrose limit} of any spacetime in a neighbourhood of any null geodesic
\cite{Penrose1976}.  This universality was revived dramatically in the string-theory
context by Berenstein--Maldacena--Nastase \cite{BMN} and Blau et al.\ \cite{Blau},
who showed that the maximally supersymmetric plane wave of type IIB supergravity is
a Penrose limit of $\mathrm{AdS}_5\times S^5$, thereby yielding an exactly solvable
model for strings in a curved background.  Standard references for the general
geometry of plane waves and their role in exact solutions of general relativity
are Stephani et al.\ \cite{Stephani} and the monograph of Blau \cite{BlauLectures}.

The present paper is the fifth in a series \cite{HS1,HS2,HS3,HS4} in which the
authors have undertaken a systematic, rigorous study of plane wave spacetimes in
arbitrary dimension, with the long-range goal of constructing and understanding the
twistor theory of these spacetimes.
Paper~I \cite{HS1} established the equivalence between the Brinkmann and Rosen
formulations, characterised the coordinate singularities of the Rosen metric as
points at which the Lagrangian curve meets a fixed Lagrangian subspace, and
related the Sachs equations governing null-geodesic congruences to the Jacobi
equation.
Paper~II \cite{HS2} classified the isometries and conformal isometries of plane waves
and exhibited families of vacuum plane waves whose isometry groups exhibit
chaotic (Bernoulli-shift) behaviour, showing that such spacetimes are not
classifiable by invariant observables.
Paper~III \cite{HS3} introduced the notion of a \emph{microcosm} (a complete,
homogeneous plane wave) and solved the Sachs equations explicitly for these
spaces, identifying the Lagrangian curve as an orbit of a one-parameter subgroup
of the symplectic group.
Paper~IV \cite{HS4} developed a general theory of cross-ratios and Schwarzians
for curves in what the authors call the \emph{middle Grassmannian}, culminating
in an intrinsic, projectively invariant formula for the curvature of a plane wave
as the Schwarzian of its Lagrangian curve.

\subsection*{Contents of this paper}

The present paper takes up the analytic side of the story: how does one
\emph{solve} the wave equation on a plane wave, and what algebraic structures
control those solutions?

\subsubsection*{The Schr\"odinger propagator and the Ward representation.}
Because the wave operator $\Box$ on a plane wave does not depend on the
retarded null coordinate $v$, a Fourier transform in $v$ reduces the wave
equation to a Schr\"odinger equation on the transverse Euclidean space
$\mathbb X$, with $u$ playing the role of time and the Laplacian
$\Delta_{G(u)}$ serving as the (time-dependent) Hamiltonian.  A second Fourier
transform, now in the transverse directions $x\in\mathbb X$, reduces this further
to an explicit quadrature, yielding the Schr\"odinger propagator of
Theorem~\ref{AuTheorem}.  The same solution can be written in the
\emph{Ward form} \cite{Ward1987}
\[
  \phi(u,v,x) = g^{-1/4}\int_{\mathbb X}
  F\!\left(v + \xi\cdot x + \tfrac{1}{2}\,\xi^T H(u)\,\xi,\,\xi\right)d^n\xi,
\]
which exhibits solutions as superpositions of progressing waves labeled by
transverse momentum $\xi$.  The tensor $H(u)$ satisfying $\dot H = G^{-1}$
is the conformally invariant datum of the plane wave; it defines the null cones
and simultaneously serves as the parameter of the Schr\"odinger representation.

\subsubsection*{The Heisenberg group.}
The isometry group of a plane wave contains a $(2n+1)$-dimensional Heisenberg
group $\mathfrak H$, which acts on the wavefronts $u = \text{const}$ by shifts
in $(v,x)$ and is responsible for the special solubility of the wave equation.
The Heisenberg group carries a family of irreducible unitary (Schr\"odinger)
representations $\rho_{h,u}$ on $L^2(\mathbb X)$, indexed by the parameter
$h\ne 0$ and depending on $u$ through $H(u)$.  The Weyl commutation relations
$[P, Q(u)] = 2\pi i h\, I$ are satisfied by the momentum operator $P$ and the
position operator $Q(u) = 2\pi x + ih H(u)\partial_x$.

The general theory of Heisenberg groups and their representations that we
require is classical, going back to the foundational work of Weyl, Stone, and
von Neumann.  For the specific combination of the Weil representation, the
Maslov index, and theta functions that arises here, the essential reference is
the monograph of Lion and Vergne \cite{LionVergne}.  Mumford's \emph{Tata
Lectures on Theta} \cite{MumfordI,MumfordIII} provides the complementary
algebraic-geometric and number-theoretic perspective on the Heisenberg group,
its theta functions, and the relationships among different polarisations.

\subsubsection*{Fourier transform on the Heisenberg group.}
A Lagrangian subspace $\mathbb X\subset\mathbb T$ of the symplectic space of
transverse positions and momenta determines a tempered distribution
$\delta_{\mathbb X}$ on the Heisenberg group.  Convolution by $\delta_{\mathbb X}$
is the abstract Fourier transform; the classical Fourier inversion theorem on
$\mathbb R^n$ is recovered as the composition
$\delta_{\mathbb X_-}*\delta_{\mathbb X_+}*\op{ind}_{\mathcal X}f = f$
for a real polarisation $\mathcal X$.  A triple convolution
$\delta_{\mathbb B}*\delta_{\mathbb C}*\delta_{\mathbb A}$ of pairwise
complementary Lagrangian distributions yields a scalar multiple of the identity,
with the phase determined by the \emph{Maslov index} $\tau(\mathbb A,\mathbb B,\mathbb C)$
of the triple \cite{LionVergne}.

\subsubsection*{The Schr\"odinger evolution and its independence of polarisation.}
The Schr\"odinger propagator $\Phi_{\mathbb X}(s,u)$ from time $s$ to time $u$
is expressed abstractly as a composition of Fourier transforms and parallel
transport along the Lagrangian curve:
\[
  \mathscr S(\mathbb X)\xrightarrow{\operatorname{ind}}
  \mathscr S(\mathbb H(s)\backslash\mathfrak H)
  \xrightarrow{\delta_{\mathbb X}*}
  \mathscr S(\mathbb X\backslash\mathfrak H)
  \xrightarrow{\Gamma_{su}}
  \mathscr S(\mathbb X\backslash\mathfrak H)
  \xrightarrow{\delta_{\mathbb H(u)}*}
  \mathscr S(\mathbb H(u)\backslash\mathfrak H)
  \xrightarrow{\operatorname{res}}
  \mathscr S(\mathbb X).
\]
The local theorem of Section~\ref{sec:schrodinger} establishes that if
$\mathbb X$ and $\mathbb X'$ are sufficiently close and both are complementary to
$\mathbb H(t)$ on a common interval, then the corresponding Schr\"odinger
evolutions are related by an explicit intertwining square:
\[
  \begin{tikzcd}
    \mathscr S(\mathbb X)\otimes\mathfrak F(s)
    \arrow[r,"\Phi_{\mathbb X}(s{,}u)"]\arrow[d,"{\rho_s}"]
    & \mathscr S(\mathbb X)\otimes\mathfrak F(u)\arrow[d,"{\rho_u}"]\\
    \mathscr S(\mathbb X')\otimes\mathfrak F(s)
    \arrow[r,swap,"\Phi_{\mathbb X'}(s{,}u)"]
    & \mathscr S(\mathbb X')\otimes\mathfrak F(u)
  \end{tikzcd}
\]
The vertical intertwining maps $\rho_s$ and $\rho_u$ depend on the
Lagrangian curve $\mathbb H(\cdot)$ and are given explicitly in terms of
quadratic exponential factors and symplectic reflection.  This local result is
the basic gluing law for the global theory.  When one works in a fixed real
polarisation, or equivalently in a fixed Rosen chart, the chart may cease to be
valid when $\mathbb H(t)$ meets its Maslov cycle: this is the appearance of a
caustic.  The Schr\"odinger evolution itself does not stop there.  Instead, one
passes to a nearby polarisation $\mathbb X'$ on an overlapping interval and uses
the local intertwiner to continue the evolution.  The global theorem is therefore
an atlas theorem: Schr\"odinger propagation continues across caustics by covering
the parameter interval with transverse polarisation charts and gluing the local
propagators by these explicit overlap maps.  The underlying Maslov phase is then
explained by the Fourier-integral calculus of Lagrangian distributions on the
Heisenberg group.

\subsubsection*{Bargmann transform and theta functions.}
For an imaginary (i.e.\ positive complex) polarisation $\mathcal J$, the
convolution kernel $\eta_{\mathcal J} = te(\mathcal J_+zz)$ is a bounded
left-$\mathcal J$-holomorphic function on $\mathfrak H$.  Convolution by
$\eta_{\mathcal J}$ defines the \emph{Bargmann transform}, which maps
Schwartz functions on $\mathfrak H$ to holomorphic Schwartz functions.  This
construction is the Heisenberg-group counterpart of the Bargmann--Segal
transform used in quantum optics and geometric quantisation.  When the Heisenberg
group is taken over a lattice (the arithmetic case relevant to modular forms),
the analogous construction produces theta functions in the sense of Mumford
\cite{MumfordI,MumfordIII}.

\subsection{Relation to the prior literature}

Wave equations on plane waves have a long history.  The progressing-wave
representation used here is due to Ward \cite{Ward1987}, who showed that
Whittaker's formula for flat-space wave solutions extends uniformly to the
plane-wave background, with the role of a plane-wave phase replaced by the
quadratic expression $v + \xi\cdot x + \tfrac{1}{2}\xi^T H(u)\xi$.

The connection between the Schr\"odinger equation and the Heisenberg group
has been extensively exploited in the study of quantum mechanics in curved
backgrounds; we particularly note the work of Duval, Burdet, K\"unzle and
Perrin \cite{DBKP} on the Bargmann structures associated to non-relativistic
spacetimes, which provides a related geometric framework.

On the representation-theoretic side, the Weil representation
(also known as the metaplectic or oscillator representation) and the Maslov
index are the central tools; the authoritative reference for our purposes is
Lion--Vergne \cite{LionVergne}.  The Lagrangian Grassmannian and its
differential geometry, particularly the Schwarzian and cross-ratio developed
in Paper~IV \cite{HS4}, underpin the global aspects of the propagator.

The theta functions appearing in the arithmetic case connect this work to a
classical tradition going back to Jacobi and Riemann.  Mumford's \emph{Tata
Lectures} \cite{MumfordI,MumfordIII} provide the modern algebraic-geometric
foundations; the relationship between theta functions and the Heisenberg group
is the organising principle of those volumes, as it is here.

\subsection{Organisation of the paper}

Section~\ref{sec:planewaves} sets up the geometry of plane waves in Rosen coordinates and derives
the Schr\"odinger propagator (Theorem~\ref{AuTheorem}) and Ward representation.
Section~\ref{sec:symplectic} introduces the Heisenberg group and its unitary representations,
establishes their action on the plane wave, and proves the group law.
Sections~\ref{sec:heisenberg}--\ref{sec:fourier} develop the abstract Fourier transform and Fourier inversion on the
Heisenberg group, including the Maslov phase theorem.
Section~\ref{sec:bargmann} discusses the Bargmann transform for imaginary polarisations and
the associated theta functions.
Section~\ref{sec:schrodinger} studies the Schr\"odinger evolution $\Phi_{\mathbb X}(s,u)$ and proves
the local commutativity theorem for nearby polarisations together with the
global continuation theorem across caustics obtained by gluing local charts.

\section{Plane waves}
\label{sec:planewaves}
We shall take a {\em plane wave} to be a spacetime manifold $\mathbb M=\mathbb U\times\mathbb R\times \mathbb X$ where $\mathbb U$ is a real interval, $\mathbb R$ denotes the real line, and $\mathbb X$ is a fixed $n$-dimensional real Euclidean space, together with the (Rosen) metric
\[\mathcal R(G) = 2\,du\,dv - dx^TG(u)dx,\]
for $(u,v,x)\in \mathbb M$, where $G$ maps $\mathbb U$ to the cone of positive-definite symmetric matrices on $\mathbb X$.  We assume that $G$ is at least $C^1$.

The wave operator of a smooth function $\phi$ is given by
\begin{equation}\label{waveoperator}\index{Wave equation}
  \Box \phi = \partial_v\partial_u\phi + \frac{1}{\sqrt{g}}\partial_u(\sqrt{g}\partial_v\phi)- \Delta\phi = 2g^{-1/4}\partial_v\partial_u(g^{1/4}\phi) - g^{-1/4}\Delta (g^{1/4}\phi)
\end{equation}
where $g=\det G(u)$, $\Delta = \partial_x G(u)^{-1}\partial_x^T$ is the Laplacian with respect to the metric $G(u)$ at $u$.  Because $\Box$ does not depend on $v$, we can apply separation of variables.  Defining the Fourier transform with respect to $v$ as
\[\hat f(\gamma) = \int_{\mathbb R}f(v)e^{-2\pi i \gamma v}\,dv,\]
the Fourier coefficients of a solution to $\Box\phi = 0$ satisfy
\begin{equation}
  \label{SchrodingerEquation}
    4\pi i \gamma (g^{1/4}\hat\phi)_u - \Delta (g^{1/4}\hat\phi)=0  % This had been +!
  \end{equation}
  which is a Schr\"odinger equation with parameter $\gamma$.  Thus $\Delta$ serves as the ``time''-dependent Hamiltonian, being a function of $u$.
  
  The conformal geometry of $\mathbb M$ is expressed most naturally, not in terms of $G$, but in terms of a function $H(u)$ taking values in symmetric matrices on $\mathbb X$, such that $dH(u)=G(u)^{-1}du$, i.e., $\dot H=G^{-1}$ where the dot is differentiation with respect to $u$.  This defining relation can be expressed as the projectively-invariant relation $\eth_0 H = G^{-1}$, and $H$ carries projective weight zero, where $\eth_0$ is the eth operator, discussed in a later paper in this series.  The tensor $H(u)$ defines the null cones in the plane wave, and also the symmetries of the spacetime.  It also determines in a natural way a (conformally invariant) curve in the Lagrangian Grassmannian of the plane wave.  We shall show how to solve \eqref{SchrodingerEquation}, using the tensor $H(u)$.    

  A second Fourier transform is relevant in the separation of variables analysis, namely that with respect to the Euclidean space $\mathbb X$.  Let $\mathscr S(\mathbb X)$ be the space of Schwartz functions on $\mathbb X$, and for $f\in\mathscr S(\mathbb X)$, put
  \[\mathcal F f(\xi) = \int_{\mathbb X} f(x)e^{-2\pi i \xi.x}\,d^nx\]
  where $d^nx$ is the normalized Lebesgue measure and $\xi.x=\xi^Tx$ is the (fixed) Euclidean inner product on $\mathbb X$.  We also have the inverse transform
  \[\mathcal F^{-1} \hat f(x) = \int_{\mathbb X} \hat f(\xi)e^{2\pi i \xi.x}\,d^n\xi.\]
  
  We then have the Schr\"odinger propagator:
  \begin{theorem}\label{AuTheorem}
    The Schwartz solution to \eqref{SchrodingerEquation} with initial data $\hat\phi_0=\hat\phi(u_0)$ is, for $\gamma\ne 0$,
    $$\hat\phi(u) = g(u)^{-1/4} \mathcal F_\xi^{-1}\!\left\{ \exp\left[\,\pi i \gamma^{-1}\,\xi^T(H(u)\!-\!H(u_0))\,\xi\,\right]\, \mathcal F\hat\phi_0(\xi)\right\}.$$
  \end{theorem}
%     4\pi i \gamma (g^{1/4}\hat\phi)_u + \Delta (g^{1/4}\hat\phi)=0
  \begin{proof}
    We have the relation $\mathcal F(\partial_xf)(\xi) = 2\pi i \xi \mathcal Ff(\xi)$, and so
    \[\mathcal F(\partial_xG(u)^{-1}\partial_x^Tf)(\xi) = -4\pi^2 \xi^T G(u)^{-1}\xi \mathcal Ff(\xi).\]
    But also,
    \begin{align*}
      4\pi i \gamma (g^{1/4}\hat\phi)_u &= -4\pi^2\mathcal F_\xi^{-1}\!\left\{ \exp\left[\pi i \gamma^{-1}\xi^T(H(u)-H(u_0))\xi\right] \,\xi^TG(u)^{-1}\xi\,\mathcal F\hat\phi_0(\xi)\right\}.
    \end{align*}
  \end{proof}
  Note that the propagator in the theorem simplifies slightly under the assumption $H(u_0)=0$, which we shall henceforth assume in this section.

  A Schwartz solution to the wave equation may be put in the Ward form \cite{Ward1987}:
  \begin{equation}\label{wardsol}
    \phi(u,v,x) = g^{-1/4}\int_{\mathbb X}F(v+\xi.x + \tfrac12 \xi^TH(u)\xi,\xi)\,d^n\xi
  \end{equation}
  where $F(v,\xi)$ is a Schwartz function.  Recalling that the hat denotes the Fourier transform in the $(v,\gamma)$ dual variables, put
  \[\phi(u,v,x) = \int_{\mathbb R} \hat \phi(u,\gamma,x)e^{2\pi i \gamma v}\,d\gamma\]
  \[F(v,\xi) = \int_{\mathbb R} \hat F(\gamma,\xi)e^{2\pi i \gamma v}\,d\gamma.\]
  For a fixed $\gamma$, we thus consider \eqref{wardsol} on the individual Fourier components $F(v,\xi)=\hat F(\gamma,\xi)e^{2\pi i \gamma v}$ and $\phi(u,v,x) = \hat\phi(u,\gamma,x)e^{2\pi i \gamma v}$, we have
  \begin{align*}
    e^{2\pi i\gamma v}\hat\phi(u,\gamma,x)
    &= g^{-1/4}\int_{\mathbb X} \hat F(\gamma,\xi)\exp\left[2\pi i \gamma\left(v + \xi.x + \frac12 \xi^TH(u)\xi\right)\right]\,d^n\xi\\
    &= g^{-1/4}|\gamma|^{-n}e^{2\pi i\gamma v}\int_{\mathbb X} \hat F(\gamma,\gamma^{-1}\xi)\exp\left[2\pi i \xi.x + \pi i \gamma^{-1}\xi^TH(u)\xi\right]\,d^n\xi\\
  \end{align*}
  in agreement with Theorem \ref{AuTheorem}, after suitable choice of $\hat F(\gamma,\xi)$.

\subsection{Well-posedness and conservation laws}
\label{sec:weaksoln}

The Schr\"odinger propagator of Theorem~\ref{AuTheorem} and the Ward
representation~\eqref{wardsol} are stated for Schwartz initial data.  We
record here the natural Hilbert-space framework that underpins these
results, showing that the space of finite-energy solutions is complete and
that the symplectic form is conserved within it.

\subsubsection*{Function spaces}

Let $\mathbb{U}$ be a compact interval, $\mathbb{V} = \mathbb{R}$, and
$\mathbb{M} = \mathbb{U}\times\mathbb{V}\times\mathbb{X}$ the plane-wave
spacetime with Rosen metric $2\,du\,dv - dx^{T}G(u)\,dx$, where
$G\in C^{1}(\mathbb{U})$ takes values in the positive-definite symmetric
endomorphisms of $\mathbb{X}$.  Let $\mathscr{S}(\mathbb{V}\times\mathbb{X})$
denote the Schwartz space of real-valued functions on
$\mathbb{V}\times\mathbb{X}$, and let $\mathscr{D}(M)$ denote the space of
compactly supported smooth functions on a manifold $M$.

Let $(\gamma,\xi)$ be the Fourier-dual variables to $(v,x)$, so that the
Fourier transform of $f$ is
\[
  \hat{f}(\gamma,\xi) = \int_{\mathbb{V}\times\mathbb{X}}
  f(v,x)\,e^{-2\pi i(\gamma v + \xi^{T}x)}\,dv\,d^{n}x.
\]
Define the weighted $L^{2}$ space $L^{2}(|\gamma|)$ as the space of
measurable functions $\hat{f}$ on $\widehat{\mathbb{V}\times\mathbb{X}}$
for which
\[
  \|\hat{f}\|^{2}_{L^{2}(|\gamma|)}
  = \int |\gamma|\,|\hat{f}(\gamma,\xi)|^{2}\,d\gamma\,d^{n}\xi < \infty,
\]
with inner product
$\langle \hat{f},\hat{k}\rangle_{L^{2}(|\gamma|)}
 = \int |\gamma|\,\overline{\hat{f}}\,\hat{k}\,d\gamma\,d^{n}\xi$.
We also write $\langle f,k\rangle_{L^{2}}
= \int_{\mathbb{V}\times\mathbb{X}}\overline{f}\,k\,dv\,d^{n}x$
for the standard $L^{2}$ inner product.

\begin{definition}
  The \emph{energy Sobolev space} $H^{1/2,0}$ is the space of real-valued
  Borel measurable functions $f$ on $\mathbb{V}\times\mathbb{X}$ such that
  \[
    \|f\|^{2}_{H^{1/2,0}}
    := \|f\|^{2}_{L^{2}} + \|\hat f\|^{2}_{L^{2}(|\gamma|)} < \infty,
  \]
  where functions equal almost everywhere are identified.
\end{definition}

\begin{definition}
  Let $X$ be a Hilbert space and $\mathbb{U}$ a compact interval.  The
  Bochner--Sobolev space $H^{1}(\mathbb{U};X)$ consists of all
  $f\in L^{2}(\mathbb{U};X)$ for which there exists
  $f'\in L^{2}(\mathbb{U};X)$ satisfying
  \[
    \int_{\mathbb{U}}\varphi'(u)\,f(u)\,du
    = -\int_{\mathbb{U}}\varphi(u)\,f'(u)\,du
    \qquad\text{for all }\varphi\in\mathscr{D}(\operatorname{int}\mathbb{U}).
  \]
  This $f'$ is called the \emph{weak derivative} of $f$, and
  $H^{1}(\mathbb{U};X)$ carries the Hilbert norm
  $\|f\|^{2} = \|f\|^{2}_{L^{2}(\mathbb{U};X)}
              + \|f'\|^{2}_{L^{2}(\mathbb{U};X)}$.
\end{definition}

The natural solution space for the wave equation is
$H^{1}(\mathbb{U};\,H^{1/2,0})$, normed by
\[
  \|f\|^{2}
  = \int_{\mathbb{U}}\bigl(
      \|f(u,\cdot)\|^{2}_{H^{1/2,0}}
    + \|f'(u,\cdot)\|^{2}_{H^{1/2,0}}
    \bigr)\sqrt{g(u)}\,du,
\]
where $g(u) = \det G(u)$.  A standard embedding theorem
(see Evans~\cite{Evans}, p.~286) gives the following continuity in $u$:

\begin{lemma}\label{lem:continuous}
  If $f\in H^{1}(\mathbb{U};\,H^{1/2,0})$, then, after modification on a
  set of measure zero, $f\in C(\mathbb{U};\,H^{1/2,0})$.  In particular,
  the slice $f(u)\in H^{1/2,0}$ is well-defined for every $u\in\mathbb{U}$.
\end{lemma}

\subsubsection*{Weak solutions}

\begin{definition}\label{def:weaksolution}
  A function $f\in H^{1}(\mathbb{U};\,H^{1/2,0})$ is a \emph{weak solution}
  of the wave equation if
  \begin{equation}\label{eq:weakwave}
    \int_{\mathbb{M}}\Bigl[
      2\,(g^{1/4}f)_{u}\,(g^{1/4}\phi)_{v}
      - g^{ij}(u)\,\partial_{i}(g^{1/4}f)\,\partial_{j}(g^{1/4}\phi)
    \Bigr]\,du\,dv\,d^{n}x = 0
  \end{equation}
  for all $\phi\in\mathscr{D}(\operatorname{int}\mathbb{M})$.  Here
  $\partial_{i}f = \partial f/\partial x^{i}$ is taken in the sense of
  distributions.
\end{definition}

Since~\eqref{eq:weakwave} is linear in the test function, a standard
Fubini argument shows that for a weak solution, the slice
condition~\eqref{eq:weakwave} holds for almost every fixed $u$:

\begin{corollary}
  If $f\in H^{1}(\mathbb{U};\,H^{1/2,0})$ is a weak solution, then
  \[
    \int_{\mathbb{V}\times\mathbb{X}}\Bigl[
      2\,(g^{1/4}f)_{u}\,(g^{1/4}\phi)_{v}
      - g^{ij}(u)\,\partial_{i}(g^{1/4}f)\,\partial_{j}(g^{1/4}\phi)
    \Bigr]\,dv\,d^{n}x = 0
  \]
  for almost every $u\in\mathbb{U}$ and all
  $\phi\in\mathscr{D}(\operatorname{int}\mathbb{M})$.
\end{corollary}

The Fourier transform converts~\eqref{eq:weakwave} into an ODE in $u$
and simultaneously yields a gradient estimate showing that weak solutions
have square-integrable transverse gradients:

\begin{proposition}\label{prop:gradest}
  Let $f\in H^{1}(\mathbb{U};\,H^{1/2,0})$ be a weak solution, and let
  $\nabla f$ denote its weak gradient in the $x$ variables.  Then for
  almost every $u\in\mathbb{U}$,
  $\nabla f(u,\cdot)\in L^{2}(\mathbb{V}\times\mathbb{X})$ and
  \[
    \|\nabla f(u,\cdot)\|^{2}_{L^{2}}
    \le C\bigl(
      \|f(u,\cdot)\|^{2}_{H^{1/2,0}}
    + \|f'(u,\cdot)\|^{2}_{H^{1/2,0}}
    \bigr).
  \]
\end{proposition}

\begin{proof}
  In the Fourier domain, a weak solution satisfies the first-order ODE
  (for almost every $u$):
  \begin{equation}\label{eq:WaveInFourier}
    4\pi i\gamma\,\widehat{(g^{1/4}f)}\,' + 4\pi^{2}\,\xi^{T}G(u)^{-1}\xi\,
    \widehat{g^{1/4}f} = 0.
  \end{equation}
  By hypothesis $\hat{f},\hat{f}'\in L^{2}(|\gamma|)$.  Multiplying
  \eqref{eq:WaveInFourier} by $\overline{\hat{f}}$, integrating, and
  applying the Cauchy--Schwarz inequality gives
  \[
    \int \xi^{T}G(u)^{-1}\xi\,|\hat{f}|^{2}\,d\gamma\,d^{n}\xi
    \le C\bigl(\|\hat{f}\|^{2}_{L^{2}(|\gamma|)}
             + \|\hat{f}'\|^{2}_{L^{2}(|\gamma|)}\bigr),
  \]
  and the left-hand side equals $\|\nabla f(u,\cdot)\|^{2}_{L^{2}}$
  by Plancherel's theorem and positive-definiteness of $G(u)$.
\end{proof}

\subsubsection*{Conservation laws}

The plane wave carries a $u$-dependent symplectic form on solutions.
For $\phi,\psi\in\mathscr{S}(\mathbb{M})$, set
\begin{equation}\label{eq:symplecticform}
  \omega_{u}(\phi,\psi)
  = \int_{\mathbb{V}\times\mathbb{X}}
    \bigl[\phi\,\psi_{v} - \phi_{v}\,\psi\bigr]
    \sqrt{g(u)}\,dv\,d^{n}x,
\end{equation}
which may be written invariantly as
$\omega_{u_{0}}(\phi,\psi)
 = \int_{u=u_{0}} {*}(\phi\,d\psi - \psi\,d\phi)$,
where $*$ is the Hodge star and the integral is over the wavefront
$\{u=u_{0}\}$.  For classical solutions the conservation
$d\omega_{u}/du = 0$ follows immediately from the wave equation via
integration by parts.  The following proposition extends this to weak
solutions.

\begin{proposition}\label{prop:conservation}
  Let $f, k\in H^{1}(\mathbb{U};\,H^{1/2,0})$ be weak solutions of the
  wave equation.  Then:
  \begin{enumerate}
  \item[\emph{(i)}] $\sqrt{g(u)}\,\langle f(u,\cdot),k(u,\cdot)\rangle_{L^{2}}$
    is independent of $u\in\mathbb{U}$;
  \item[\emph{(ii)}] $\omega_{u}(f,k)$ is independent of $u\in\mathbb{U}$.
  \end{enumerate}
\end{proposition}

The proof requires three lemmas about weak derivatives in Bochner spaces.

\begin{lemma}[Difference quotient approximation]\label{lem:diffquot}
  Let $X$ be a Hilbert space and $\phi\in H^{1}(\mathbb{U};X)$.  Define
  \[
    D^{\epsilon}\phi(u) =
    \begin{cases}
      \dfrac{\phi(u+\epsilon)-\phi(u)}{\epsilon}
        & \text{if }u,\,u+\epsilon\in\mathbb{U},\\[4pt]
      0 & \text{otherwise.}
    \end{cases}
  \]
  Then $D^{\epsilon}\phi \rightharpoonup \phi'$ weakly in
  $L^{2}(\mathbb{U};X)$ as $\epsilon\to 0$.
\end{lemma}

\begin{proof}
  By the fundamental theorem of calculus and Cauchy--Schwarz,
  \[
    \|D^{\epsilon}\phi(u)\|_{X}
    \le \frac{1}{\epsilon}\int_{0}^{\epsilon}\|\phi'(u+t)\|_{X}\,dt
    \le \frac{1}{\sqrt{\epsilon}}
        \Bigl(\int_{0}^{\epsilon}\|\phi'(u+t)\|_{X}^{2}\,dt\Bigr)^{1/2}.
  \]
  Squaring and integrating over $\mathbb{U}$ via Fubini gives
  $\|D^{\epsilon}\phi\|_{L^{2}(\mathbb{U};X)}^{2}
   \le \|\phi'\|_{L^{2}(\mathbb{U};X)}^{2}$ for all $\epsilon > 0$,
  so the family $\{D^{\epsilon}\phi\}$ is bounded in $L^{2}(\mathbb{U};X)$.
  By the Banach--Alaoglu theorem some subnet converges weakly to a limit
  $\psi$; since $D^{\epsilon}\phi\to\phi'$ in the sense of
  $X$-valued distributions, we conclude $\psi = \phi'$.  As every
  weakly convergent subnet has the same limit, $D^{\epsilon}\phi
  \rightharpoonup \phi'$.
\end{proof}

\begin{lemma}[Weak--strong product]\label{lem:weakstrong}
  If $\phi_{\epsilon}\to\phi$ strongly in $H^{s}$ and
  $\psi_{\epsilon}\rightharpoonup\psi$ weakly in $H^{-s}$, then
  $\phi_{\epsilon}\psi_{\epsilon}\rightharpoonup\phi\psi$ weakly in $L^{1}$.
\end{lemma}

\begin{proof}
  For any $f\in L^{\infty}$, write
  \[
    \int f(\phi_{\epsilon}\psi_{\epsilon}-\phi\psi)
    = \int f(\phi_{\epsilon}-\phi)\psi_{\epsilon}
    + \int f\phi(\psi_{\epsilon}-\psi).
  \]
  The first term is bounded by
  $\|f\|_{\infty}\|\psi_{\epsilon}\|_{H^{-s}}\|\phi_{\epsilon}-\phi\|_{H^{s}}
  \to 0$,
  since weakly convergent sequences are bounded (uniform boundedness
  principle) and $\phi_{\epsilon}\to\phi$ strongly.  The second term
  tends to zero because $\psi_{\epsilon}\rightharpoonup\psi$ and
  $f\phi\in H^{s}$.
\end{proof}

\begin{lemma}[Product rule for weak derivatives]\label{lem:productrule}
  Let $\phi\in H^{1}(\mathbb{U};\,H^{1/2,0})$ and
  $\psi\in H^{1}(\mathbb{U};\,H^{-1/2,0})$.  Then
  $(\phi\psi)'\in L^{1}(\mathbb{M})$ and
  $(\phi\psi)' = \phi'\psi + \phi\psi'$ almost everywhere.
\end{lemma}

\begin{proof}
  The right-hand side belongs to $L^{1}(\mathbb{M})$ by the hypotheses
  and duality.  For the product rule, use the discrete product formula
  $D^{\epsilon}(\phi\psi)(u)
   = D^{\epsilon}\phi(u)\cdot\psi(u+\epsilon) + \phi(u)\cdot D^{\epsilon}\psi(u)$.
  As $\epsilon\to 0$, the left side converges weakly to $(\phi\psi)'$
  (in the sense of $L^{1}(\mathbb{M})$) by Lemma~\ref{lem:diffquot},
  while both terms on the right converge weakly to $\phi'\psi$ and
  $\phi\psi'$ respectively by Lemma~\ref{lem:weakstrong}.
\end{proof}

\begin{proof}[Proof of Proposition~\ref{prop:conservation}]
  \textit{Part~(ii).}  We show that the weak derivative in $u$ of
  $\omega_{u}(f,k)$ vanishes.  For any
  $\varphi\in\mathscr{D}(\operatorname{int}\mathbb{U})$, Fubini's theorem gives
  \begin{align*}
    \int_{\mathbb{U}}\varphi'(u)\,\omega_{u}(f,k)\,du
    &= 2\int_{\mathbb{M}}
       \varphi'(u)\,f\,k_{v}\,\sqrt{g(u)}\,du\,dv\,d^{n}x.
  \end{align*}
  Integrating by parts in $u$ (justified by
  Lemmas~\ref{lem:productrule} and~\ref{lem:weakstrong}) and then in $v$,
  and using the weak wave equation~\eqref{eq:weakwave} for both $f$ and
  $k$, each resulting boundary term cancels its counterpart:
  \begin{align*}
    &= -2\int_{\mathbb{M}}
       \varphi(u)\bigl[(g^{1/4}f)_{u}(g^{1/4}k)_{v}
                      +(g^{1/4}f)(g^{1/4}k)_{uv}\bigr]\,du\,dv\,d^{n}x \\
    &= 2\int_{\mathbb{M}}
       \varphi(u)\bigl[(g^{1/4}f)_{uv}(g^{1/4}k)
                      -(g^{1/4}f)(g^{1/4}k)_{uv}\bigr]\,du\,dv\,d^{n}x \\
    &= 2\int_{\mathbb{M}}
       \varphi(u)\bigl[\Delta(g^{1/4}f)\cdot g^{1/4}k
                      - g^{1/4}f\cdot\Delta(g^{1/4}k)\bigr]
       \,du\,dv\,d^{n}x
    = 0,
  \end{align*}
  where the last step uses self-adjointness of the Laplacian $\Delta$.

  \textit{Part~(i).}  In the Fourier domain, using
  equation~\eqref{eq:WaveInFourier} for $f$ and $k$:
  \begin{align*}
    \frac{d}{du}\bigl(\sqrt{g}\,\langle f,k\rangle_{L^{2}}\bigr)
    &= \int \overline{(g^{1/4}\hat{f})'}\,(g^{1/4}\hat{k})
         + \overline{(g^{1/4}\hat{f})}\,(g^{1/4}\hat{k})'\,d\gamma\,d^{n}\xi \\
    &= \int \frac{i\xi^{T}G^{-1}\xi}{\gamma}
       \bigl(
         -\overline{g^{1/4}\hat{f}}\cdot g^{1/4}\hat{k}
         +\overline{g^{1/4}\hat{f}}\cdot g^{1/4}\hat{k}
       \bigr)\,d\gamma\,d^{n}\xi
    = 0,
  \end{align*}
  since $\gamma$ and $\xi$ are real so the integrand is purely imaginary
  and its real part vanishes. \qedhere
\end{proof}

\begin{remark}
  Proposition~\ref{prop:conservation}(ii) is the weak-solution
  counterpart of the formal calculation at the start of this section.
  Together with Lemma~\ref{lem:continuous}, it implies that the
  Schr\"odinger flow $f(u_0)\mapsto f(u)$ extends from Schwartz initial
  data to a one-parameter family of unitary isomorphisms on the Hilbert
  space of finite-energy solutions, completing the functional-analytic
  justification of the propagator formula of Theorem~\ref{AuTheorem}.
\end{remark}
  %%%%

  \subsection{Heisenberg group}

Let $L^2(\mathbb X)$ denote the complex Hilbert space of square-integrable functions with respect to the normalized Lebesgue measure of the Euclidean space $\mathbb X$.  We shall use coordinates $x$ on $\mathbb X$ and $\xi$ on the dual space, such that the canonical one-form of $T^*\mathbb X$ is
$$\theta = \xi.dx.$$
Let $u$ be a real parameter and $H(u)$ a symmetric bilinear form on $\mathbb X$, smooth for $u\in\mathbb U$, such that $dH/du$ is positive-definite.

The Fourier transform is a unitary mapping from $L^2(\mathbb X)$ to itself, that is given on the dense subset of $L^1\cap L^2$ by the formula
$$\mathcal F f(\xi) = \int_{\mathbb X} e^{-2\pi i \xi\cdot x} f(x)\,dx$$
The inverse Fourier transform is then
$$\mathcal F^{-1} g(x) = \int_{\mathbb X} e^{2\pi i \xi\cdot x} g(\xi)\,d\xi.$$

The Schwartz space $\mathscr S(\mathbb X)$ is the set of functions tempered by polynomials in both Fourier domains.  It is well-known that $\mathscr S(\mathbb X)$ is dense in $L^2(\mathbb X)$.  We observe the following standard:
\begin{lemma}
  Let $f\in\mathscr S(\mathbb X)$.  Then
  $$\mathcal F(\partial_x f)(\xi) = 2\pi i \xi \mathcal F f(\xi).$$
\end{lemma}
% \begin{proof}
%   Integration by parts, being justified by the rapid decay of the integrand at infinity, gives
%   \begin{align*}
%     \mathcal F(\partial_i f)(\xi) &= \int_{\mathbb X} e^{-2\pi i \xi\cdot x}\partial_if(x)\,dx\\
%                                   &= -\int_{\mathbb X} \partial_ie^{-2\pi i \xi\cdot x}f(x)\,dx\\
%                                   &= 2\pi i \xi_i \int_{\mathbb X} e^{-2\pi i \xi\cdot x}f(x)\,dx,
%   \end{align*}
%   as claimed.
% \end{proof}

Acting on the set of Schwartz functions in $L^2(\mathbb X)$, define operators $P,Q(u)$ by
\begin{align*}
  P f(x) &= ih\partial_x f(x)\\
  Q(u) f(x) &= (2\pi x + ihH(u)\partial_x)f(x).
\end{align*}

On the momentum domain, one then has
\begin{align*}
  P g(\xi) &= -2\pi h\xi g(\xi)\\
  Q(u) g(\xi) &= (i \partial^\xi - 2\pi hH(u)\xi)g(\xi).
\end{align*}
These operators are self-adjoint unbounded operators on $L^2(\mathbb X)$, satisfying the canonical Weyl commutation relations
$$[P,Q(u)] = 2\pi i h I$$
where $I$ is the identity operator.  Thus the operators give a family of infinitesimal (Schr\"odinger) representations of the Heisenberg algebra on $L^2(\mathbb X)$, indexed by the parameter $h\ne 0$.

\subsubsection{Integrated representation}

We now consider the associated integrated unitary representations of the Heisenberg groups.  The Heisenberg group $\tilde{\mathfrak H}$ shall consist of triples $(z,q,p)\in\mathfrak H = \mathbb R\times\mathbb X\times\mathbb X$, with multiplication
\begin{equation}\label{HGroupLaw}
  (z,p,q)(z',p',q') = \left(z + z' + \frac12(q.p' - p.q'), p+p', q+q'\right).
\end{equation}
\begin{lemma}
  The multiplication law \eqref{HGroupLaw} gives $\tilde{\mathfrak H}$ the structure of a group in which
  $$(z,p,q)^{-1} = (-z,-p,-q) $$
  and for which $(0,0,0)$ is identity.
\end{lemma}
\begin{proof}
  Only the associative law is not obvious.  We have
  \begin{align*}
    \left(\, (z,p,q)(z',p',q')\, \right)(z'',p'',q'') &= (z + \tfrac12(q p'-p q') + \tfrac12( (q + q') p'' -(p + p') q'') + z' + z'', \\
    &\qquad\qquad p + p' + p'', q + q' + q'')
  \end{align*}
  and
  \begin{align*}
  (z,p,q)\left(\, (z',p',q')(z'',p'',q'')\, \right) &= (z + \tfrac12(q (p' + p'')-p (q' + q'') ) + \tfrac12(q' p'' -p' q'') + z' + z'',\\
  &\qquad\qquad p + p' + p'', q + q' + q'').
  \end{align*}
  Comparing terms, these are readily seen to be equal.
\end{proof}

Let $h\not=0$ be a fixed real constant.  Consider the action of $\tilde{\mathfrak H}$ on $L^2(\mathbb X)$ given by
\begin{equation}\label{UnitaryRepOnL2}
  \rho_{h,u}(z,p,q)f(x) = e\left(2hz + 2x^ip_i + h(q.p - p^TH^{ij}(u)p)\right)\,f(x-hH(u)p + hq)
\end{equation}

\begin{lemma}
  The definition \eqref{UnitaryRepOnL2} is a unitary action of the group $\tilde{\mathfrak H}$ on $L^2(\mathbb X)$ with infinitesimal generators $-iP$ and $iQ$.
\end{lemma}
\begin{proof}
  Unitarity follows from the fact that $e_h(z) = \exp(2\pi i h z) = e(2hz)$ is a unitary character (all of the variables are real), together with the translation-invariance of the Lebesgue measure.  To verify that it is an action of the group, we must check that
  \begin{align*} \rho_{h,u}(z,p,q)&(\rho_{h,u}(z',p',q')f(x))\\
    &= \rho_{h,u}(z + z' + \tfrac12(q.p' - p.q'), p+p', q+q')f(x).
  \end{align*}
  Expanding the LHS gives
  \begin{align*} \rho_{h,u}(z,q,p)
    &(\rho_{h,u}(z',p',q')f(x))\\
    &=\rho_{h,u}(z,p,q)e\left(2hz' + 2x.p' + h(q'.p' - H(p',p'))\right)\,f(x-hH(p') + hq')\\
    &=e\left(2hz + 2x.p + h(q.p-H(p,p)) + 2hz' + 2(x - hH(p) + hq).p' + h(q'.p' - H(p',p'))\right)\\
    &\qquad\qquad f(x-hH(p') + hq' - hH(p) + hq)\\
    &=e\left(2hz + 2hz' + h(q.p' - p.q') + 2x.(p+p') + h((q+q').(p+p')-H(p+p',p+p'))\right)\\
    &\qquad\qquad f(x-hH(p+p') + h(q+q')).
  \end{align*}
  The final line is easily seen to be the expected RHS.

  Finally, we must determine the infinitesimal generators.  For the first generator, take $(z,p,q) = (0,0,q)$ with $q$ infinitesimal:
  $$\rho_{h,u}(0,0,q)f(x) = f(x+hq) = f(x) + hq\cdot\nabla f(x) = f(x)-i q.Pf(x).$$
  For the second generator, take $(z,p,q) = (0,p,0)$ with $p$ infinitesimal:
  \begin{align*}
    \rho_{h,u}(0,p,0)f(x) &= e(2x.p)f(x-hH(p)) = f(x) + (2\pi i x.p - h H(p).\nabla)f(x)\\
                          &= f(x) + ip.Qf(x).
  \end{align*}
\end{proof}

  \subsubsection{Action on a plane wave}\label{HeisenbergAction}
  The Heisenberg group acts by isometries on the plane wave $\mathcal R(G)=2\,du\,dv - dx^TG(u)dx$, preserving the wave fronts.  The {\em right} action on points is
    $$(u,v,x).(z,p,q) = (u, v + z + p^Tx +\tfrac12(p^Tq - p^TH(u)p), x + q - H(u)p).$$
    This defines a representation of $\tilde{\mathfrak H}$ on the set of square-integrable functions on $\mathbb R\times\mathbb X$ (at a fixed $u$).  Note that the Lebesgue measure is preserved, so it is a unitary representation.  Call this representation $\pi_u$.  After Fourier transform in the $v$-variable, the fibre at frequency $h$ is not literally $\rho_{h,u}$ but is unitarily equivalent to it via the dilation
    $$D_hf(x)=|h|^{-n/2}f(x/h).$$
    Equivalently, if $\widetilde\rho_{h,u}$ denotes the representation on the $h$-fibre coming from the geometric action, then
    $$\widetilde\rho_{h,u}=D_h^{-1}\rho_{h,u}D_h.$$ 
    Thus $\pi_u$ decomposes fibrewise into representations unitarily equivalent to the Schr\"odinger models $\rho_{h,u}$.

    We shall therefore henceforth consider the group obtained from the Heisenberg group by taking the quotient by a discrete subgroup of its center.  

% heis[{z_, p_, q_}, {z1_, p1_, q1_}] = {z + z1 + 1/2 (q . p1 - q1 . p),
%    p + p1, q + q1}
% rosact[{z_, p_, q_}, {v_, x_}] = {v - z - p . x + 
%    1/2 ( p . q - p . H . p), x - q + H . p}
% tmp1 = FullSimplify[
%   TensorExpand /@ rosact[{z, p, q}, rosact[{z1, p1, q1}, {v, x}]]]
% tmp2 = FullSimplify[
%   TensorExpand /@ rosact[heis[{z, p, q}, {z1, p1, q1}], {v, x}]]
% FullSimplify[tmp1 == tmp2 /. {Dot[a_, b__] :> Dot @@ Sort[{a, b}]}]

%\chapter{Fourier analysis}

    \section{Symplectic vector spaces}
    \label{sec:symplectic}
Let $(\mathbb T,\omega)$ denote a $2n$-dimensional symplectic vector space over the real field $\mathbb R$.  The symplectic group $\op{Sp}(\mathbb T)$ denotes the group of all symplectic linear automorphisms of $\mathbb T$, a Lie group of dimension $n(2n+1)$.  A subset $S\subset\mathbb T$ is called isotropic if the symplectic form is zero on all pairs of elements of $S$.  The maximal isotropic sets are $n$-dimensional linear subspaces of $\mathbb T$, called Lagrangian subspaces.  Let $\op{LG}(\mathbb T)$ denote the set of Lagrangian subspaces of $\mathbb T$, a smooth projective variety of dimension $n(n+1)/2$.

We consider the vector space $\mathbb T\otimes\mathbb C$ over the complex field as a symplectic space where $\omega$ is now complex-valued (skew, complex bilinear, and restricts to $\omega$ on the real $\mathbb T$).  The Lagrangian Grassmannian $\op{LG}(\mathbb T\otimes\mathbb C)$ is a complex projective (smooth) variety, and $\op{LG}(\mathbb T)$ is the set of fixed points of the Galois involution.

A pair $\mathbb A$ and $\mathbb B$ of Lagrangian subspaces are said to $k$-meet if the codimension of $\mathbb A\cap\mathbb B$ in $\mathbb A$ (or, equivalently, in $\mathbb B$) is $k$.  Thus two Lagrangian subspaces 0-meet iff they are identical.  

A polarization of $\mathbb T$ is a linear operator $\mathcal X:\mathbb T\to\mathbb T\otimes\mathbb C$ such that $[\mathcal X,\bar{\mathcal X}]=0$, $\omega(\mathcal Xx,\mathcal Xy) = -\omega(x,y)$ for all $x,y\in\mathbb T$, and $\mathcal X^2=1$ (the identity of $\mathbb T$).  A polarization $\mathcal X$ is {\em real} if its image is the real subspace $\mathbb T$ of $\mathbb T\otimes\mathbb C$.  If $i\mathcal X$ is real, then we say that $\mathcal X$ is an {\em imaginary} polarization.  A polarization $\mathcal X$ can be a mixture of real and imaginary, but because $[\mathcal X,\bar{\mathcal X}]=0$, the polarization decomposes into a real part and an imaginary part that act on complementary invariant subspaces of $\mathbb T$, so we henceforth consider only these two cases.

Given a polarization $\mathcal X$, define the projection operators $\mathcal X_\pm$ onto the respective $\pm1$ eigenspaces of $\mathcal X$, $\mathbb X_\pm$:
\[\mathcal X_{_{+}} = (1 + \mathcal X)/2,\quad \mathcal X_{_{-}} = (1-\mathcal X)/2.\]
Thus we have the usual relations
\[1 = \mathcal X_{_{+}}+\mathcal X_{_{-}},\quad \mathcal X=\mathcal X_{_{+}}-\mathcal X_{_{-}},\quad \mathcal X_\pm^2=\mathcal X_\pm, \quad \mathcal X_{_{+}}\mathcal X_{_{-}}=\mathcal X_{_{-}}\mathcal X_{_{+}}=0.\]

An imaginary polarization $\mathcal X$ is called {\em positive} if $i\omega(x,\mathcal Xx)>0$ for all $x\in\mathbb T$, $x\ne 0$.  More generally, the signature of an imaginary polarization $\mathcal X$ is the signature of the quadratic form $i\omega(x,\mathcal Xx)$ on $\mathbb T$.  The signature of a quadratic form is a pair of integers $(k_{_{+}},k_{_{-}})$ being the number of $\pm 1$ in the Sylvester normal form.  The (inertial) index of a quadratic form is $k_{_{+}}-k_{_{-}}$, the trace of its Sylvester normal form.  A quadratic form is said to split if it has index zero, i.e., $k_{_{+}}=k_{_{-}}$.  Accordingly, we say that an imaginary polarization $\mathcal X$ is split if the quadratic form $i\omega(x,\mathcal Xx)$ has signature $(n,n)$.

For a finite-dimensional real vector space $\mathbb A$, the space of Lebesgue measurable functions on $\mathbb A$ is denoted by $\mathscr M_r(\mathbb A)$.  If $\mathbb A$ is a finite-dimensional complex vector space, then $\mathscr M_i(\mathbb A)$ is the space of all holomorphic functions on $\mathbb A$.

\section{Heisenberg groups}
\label{sec:heisenberg}
Fix a non-trivial real character $e:\mathbb R\to U(1)$, which we take to be $e(t) = e^{\pi i t}$.   The Heisenberg group of $\mathbb T$ with character $e$ is the set $\mathfrak H(\mathbb T) = U(1) \times \mathbb T$ whose group law is
\begin{equation}\label{Hgrouplaw}
  (t,x)(t',x') = (tt'e(xx'),x+x').
\end{equation}
We always fix a left Haar measure on $\mathfrak H$, which is also a right Haar measure.  We fix this in such a way that it is the product measure of the invariant probability measure on $U(1)$ and the Liouville measure on $\mathbb T$.  The universal cover $\tilde{\mathfrak H}$ is identified with $\mathbb R\times\mathbb T$ and has group law
\[(t,x)(t',x') = (t+t' + xx',x+x').\]
As mentioned in \S\ref{HeisenbergAction}, we shall henceforth consider the group $\mathfrak H$ in lieu of its universal cover.

The Heisenberg group is equipped with an exact sequence
\[1\to U(1)\to\mathfrak H\xrightarrow{\pi}\mathbb T\to 0.\]
Let $A\subset\mathfrak H$ be a maximal abelian subgroup such that $\pi|_A:A\to\pi(A)$ is an isomorphism onto its image in $\mathbb T$.  Then $A$ is closed.  Several possibilities of $A$ are:
\begin{itemize}
\item $\pi(A)$ is a Lagrangian subspace $\mathbb A$ of $\mathbb T$, and $A$ is just $\{1\}\times\mathbb A$.
\item $\pi(A)$ is a self-dual lattice $\Lambda$ in $\mathbb T$, and $A$ is a set of pairs $(e(Q(\lambda)),\lambda)$ where $\lambda\in\Lambda$ and $Q$ is an $F_2$-valued quadratic form on the $F_2$ vector space $\Lambda/2\Lambda$, such that $Q(x+y)+Q(x)+Q(y)+\omega(x,y)\equiv 0\pmod 2$ for all $x,y\in\Lambda/2\Lambda$.
\end{itemize}
These possibilities are obviously not mutually exclusive.  Nevertheless, $\pi(A)$ is a closed abelian subgroup of the vector space $\mathbb T$, and so factorizes as the product of a lattice in lower dimension and isotropic subspace complementary to the lattice.  So these two cases are the main ones, up to decomposing things appropriately in lower dimensions.

The complexified Heisenberg group is the group $\mathfrak H_{\mathbb C}=\mathbb C^*\times\mathbb T\otimes\mathbb C$, with the same group law \eqref{Hgrouplaw}.  Here $e$ is extended to the unique analytic quasicharacter $e:\mathbb C\to\mathbb C^*$ which restricts to the character $e$ on $\mathbb R$, i.e., $e(z) = e^{\pi i z}$ for $z\in\mathbb C$.  

Let $\mathscr M_r(\mathfrak H)$ be the space of measurable functions on $\mathfrak H$ that commute with the action of the center:
\[f(t,x) = tf(1,x)\]
for all $t\in U(1)$ and $x\in\mathbb T$.  Let $\mathscr M_i(\mathfrak H)$ be the space of holomorphic functions on $\mathfrak H_{\mathbb C}$ that commute with the action of the center $\mathbb C^*$.  Henceforth, all functions we shall consider on $\mathfrak H$ or $\mathfrak H_{\mathbb C}$ satisfy the condition of commuting with the center.

If $f\in\mathscr M_r(\mathfrak H)$, then the left and right actions of an element $(t,x)\in\mathfrak H$ are given, respectively, by
\begin{align*}
  L(t,x)f(t',x') &= f((t,x)(t',x')) = tt'e(xx')f(1,x+x'),\\
  R(t,x)f(t',x') &=f((t',x')(t,x)) = tt'e(-xx')f(1,x+x').
\end{align*}
Identical formulas apply for $f\in\mathscr M_i(\mathfrak H)$.

Given a subgroup $A$ of $\mathfrak H$, let $\mathscr M_r(A\setminus\mathfrak H)$ be the set of left $A$-invariant elements of $\mathscr M_r(\mathfrak H)$: $L(z)f(x)=f(x)$ for all $z\in A$.  (Also, let $\mathscr M_r(\mathfrak H/A)$ be the set of right $A$-invariant elements of $\mathscr M_r(\mathfrak H)$.)

A Schwartz function on $\mathfrak H$ is an element of $\mathscr M_r(\mathfrak H)$ which is smooth and decays at infinity with all derivatives.  The convolution of Schwartz functions $f,g$ on $\mathfrak H$ is defined by
\[(f*g)(x) = \int_{\mathfrak H}f(y)g(y^{-1}x)dy\]
The convolution of Schwartz functions is Schwartz. Moreover, $L(x)(f*g) = (L(x)f)*g$ and $R(x)(f*g) = f*(R(x)g)$.  For a tempered distribution $T$ and a function $f$, the convolutions $T*f$ and $f*T$ are defined by duality against a Schwartz ``test function'' $\phi$:
\begin{align*}
  \langle T*f, \phi\rangle &= \int T(y)f(y^{-1}x)dy\,\phi(x)dx = \int f(y^{-1}x)\phi(x)dx\,T(y) dy\\
                           &=\int f(y^{-1}x)\phi(x)dx\,T(y) dy = \int (\phi * \tilde f)(y) T(y)dy\\
                           &=\langle T, \phi * \tilde f\rangle.\\
  \langle f*T, \phi\rangle &= \int f(y)T(y^{-1}x)dy\,\phi(x)dx = \int f(xy)T(y^{-1})dy\,\phi(x)dx\\
                           &= \int \phi(x)f(xy)dx\,T(y^{-1})dy\\
                           &=\langle \tilde T, (\tilde f * \tilde \phi)^\sim\rangle.
\end{align*}
Here we denote by $\tilde f$ (or $f^\sim$) the function $\tilde f(x) = f(x^{-1})$.

For an imaginary polarization $\mathcal J$, a function $f\in\mathscr M_i(\mathfrak H)$ is called left $\mathcal J$-holomorphic if $f$ is left $\mathbb J_{_{-}}$-invariant.  The space of left $\mathcal J$-holomorphic functions is denoted by
\[\mathscr M(\mathbb J_{_{-}}\setminus\mathfrak H).\]
A left $\mathcal J$-holomorphic function is determined by its restriction to the real group $\mathfrak H$.

Given a polarization $\mathcal X$, we have a map from $\mathscr M(\mathbb X_{_{+}})$ to $\mathscr M(\mathbb X_{_{-}}\setminus\mathfrak H)$, as follows.  Let $f\in\mathscr M(\mathbb X_{_{+}})$, and define $\op{ind}_{\mathcal X}f\in\mathscr M(\mathbb X_{_{-}}\setminus\mathfrak H)$ by
\[\op{ind}_{\mathcal X}f(t,z) = te(\omega(\mathcal X_{_{+}}z,z))f(\mathcal X_{_{+}}z).\]
(This definition makes sense in both $\mathscr M_r$ and $\mathscr M_c$.)

As a special case, we have $\op{ind}_{\mathcal X}1(t,z) = te(\omega(\mathcal X_{_{+}}z,z))$.  When $\mathcal X$ is a positive imaginary polarization, $\op{ind}_{\mathcal X}1$ is Schwartz on $\mathbb T$ and left $\mathcal X$-holomorphic.  Left holomorphicity of $f\in\mathscr S(\mathfrak H)$ is characterized by a Cauchy--Riemann differential equation $df\circ\mathcal X_{_{-}}=0$.  Writing $z=z_{_{+}}+z_{_{-}}$, $z_\pm=\mathcal X_\pm z\in\mathbb X_\pm$, this equation is $\partial_{_{-}} f(t,z) = i \pi f(t,z)\omega(z_{_{+}},dz_{_{-}})$.
Thus:
\begin{itemize}
\item For an imaginary polarization $\mathcal J$, the space $\mathscr S(\mathbb J_{_{-}}\setminus\mathfrak H)$ is the set of Schwartz functions on $\mathbb T$ that are left $\mathcal J$-holomorphic.
\item For a real polarization $\mathcal K$, the space $\mathscr S(\mathbb K_{_{-}}\setminus\mathfrak H)$ is the set of left $\mathcal K_{_{-}}$-invariant functions whose restrictions to $\mathbb K_{_{+}}$ are Schwartz.
\end{itemize}

\section{Fourier transform}
\label{sec:fourier}
Let $\mathbb X$ be a Lagrangian subspace.  Define a tempered distribution $\delta_{\mathbb X}$ on $\mathfrak H$ by
\begin{align*}
\langle \delta_{\mathbb X}, f\rangle = \int_{\mathbb X} f(1,x)\,dx.
\end{align*}
The measure on the right-hand side is an arbitrary Haar measure on $\mathbb X$.  This is a continuous linear functional in the Schwartz space.  The distribution $\delta_{\mathbb X}$ is $\mathbb X$-biinvariant.  More generally, the same $\delta_{\mathbb X}$ defines an $\mathbb X$-biinvariant tempered distribution on any $\mathscr S(\mathbb A\setminus\mathfrak H)$ where $\mathbb A$ is any isotropic subspace independent from $\mathbb X$.

Because it depends on a Haar measure on $\mathbb X$ (which is not fixed by the Liouville measure on $\mathbb T$), $\delta_{\mathbb X}$ takes values in the dual space to the space of Haar measures on $\mathbb X$ (a one dimensional space of densities).

\begin{definition}
  The Fourier transform of a function $f$ on $\mathfrak H$ with respect to $\mathbb X$ is the convolution
  \[\delta_{\mathbb X}*f\]
  which makes sense provided $f\in\mathscr S(\mathfrak H)$ or $f\in\mathscr S(A\setminus\mathfrak H)$ where $A$ is independent of $\mathbb X$.
\end{definition}
\begin{theorem}
Let $\mathbb X$ be a Lagrangian subspace of $\mathbb T$.  Then $\delta_{\mathbb X}*f$ belongs to $\mathscr S(\mathbb X\setminus\mathfrak H)$ for all $f\in\mathscr S(\mathbb A\setminus\mathfrak H)$, where $\mathbb A$ is an isotropic subspace independent of $\mathbb X$.
\end{theorem}
\begin{proof}
  Left invariance follows by what we have already discussed.  We need to show that $\delta_{\mathbb X}*f$ is Schwartz.  Enlarging $\mathbb A$ to a Lagrangian complement, we must therefore show that $z\mapsto (\delta_{\mathbb X}*f)(1,z)$ is a Schwartz function for $z$ belonging to $\mathbb A$.  Let $\mathcal X$ be the polarization with eigenspaces $\mathbb X_{_{+}}=\mathbb X$ and $\mathbb X_{_{-}}=\mathbb A$.  We then have $f(t,z) = e(\mathcal X_{_{+}}zz)\phi_{_{+}}(\mathcal X_{_{+}}z)$ where $\phi$ is a Schwartz function on $\mathbb X_{_{+}}$.  The Fourier transform is
  \begin{align*}
    (\delta_{\mathbb X}*f)(z)
    &= \int_{\mathfrak H}\delta_{\mathbb X}(y)f(y^{-1}z)\,dy\\
    &= \int_{\mathbb X}f(y^{-1}z)dy = \int_{\mathbb X}e(zy_{_{+}} + (z_{_{+}}-y_{_{+}})(z-y_{_{+}}))\phi(z_{_{+}}-y_{_{+}})\,dy\\
    &= \int_{\mathbb X}f(y^{-1}z)dy = \int_{\mathbb X}e(zy_{_{+}} -y_{_{+}}z))\phi(-y_{_{+}})\,dy\quad (z_{_{+}}=0)\\
    &= \int_{\mathbb X}e(2y_{_{+}}z)\phi(y_{_{+}})\,dy
  \end{align*}
  i.e., the classical Fourier transform of $\phi$.
\end{proof}

A corollary to the above proof is that the Fourier transform $f\mapsto \delta_{\mathbb X}*f$ is unitary with respect to the natural $L^2$ norms on $\mathscr S(\mathbb X\setminus\mathfrak H)$ and $\mathscr S(\mathbb A\setminus\mathfrak H)$.  Moreover, it intertwines the (unitary) right actions of the Heisenberg group (i.e., the Schr\"odinger representations).

Suppose that $\mathcal X$ is a real polarization of $\mathbb T$.  Consider the composite
\[\begin{tikzcd}
    \mathscr S(\mathbb X_{_{+}}) \arrow[r,"\op{ind}_{\mathcal X}"] & \mathscr S(\mathbb X_{_{-}}\setminus\mathfrak H)\arrow[r,"\delta_{\mathbb X_{_{+}}}\!*"] & \mathscr S(\mathbb X_{_{+}}\setminus\mathfrak H)\arrow[r,"\delta_{\mathbb X_{_{-}}}\!*"] & \mathscr S(\mathbb X_{_{-}}\setminus\mathfrak H)
  \end{tikzcd}
\]
We then have Fourier inversion:
\begin{lemma}
  If $\mathcal X$ is a real polarization of $\mathbb T$, then for all $f\in\mathscr S(\mathbb X_{_{-}}\setminus\mathfrak H)$,
  \[\delta_{\mathbb X_{_{-}}}\!\!*\delta_{\mathbb X_{_{+}}}\!\!*f = f\]
  where the Lebesgue measures on $\mathbb X_{_{-}}$ and $\mathbb X_{_{+}}$ are chosen such that the product measure is the Liouville measure on $\mathbb T$.
\end{lemma}
(We can write the conclusion of the lemma more invariantly in terms of Haar measures $\mu_{\mathbb X_\pm}$ on $\mathbb X_{_{\pm}}$: $(\mu_{\mathbb X_{_{-}}}\otimes \mu_{\mathbb X_{_{+}}})\cdot\delta_{\mathbb X_{_{-}}}\!\!*\delta_{\mathbb X_{_{+}}}\!\!*f = \frac{d\mathcal \mu_{\mathbb X_{_{-}}}d\mu_{\mathbb X_{_{+}}}}{d\mathcal L}\cdot f$ where $d\mathcal \mu_{\mathbb X_{_{-}}}d\mu_{\mathbb X_{_{+}}}/d\mathcal L$ is the (constant!) Radon--Nikodym derivative of the product measure $\mu_{\mathbb X_{_{-}}}\times\mu_{\mathbb X_{_{+}}}$ with respect to the Liouville measure $\mathcal L$ on $\mathbb T$.)
\begin{proof}
  We write $x=x_{_{+}}+x_{_{-}}$ as the decomposition of $\mathbb T$ into eigenspace components $x_\pm\in\mathbb X_\pm$.  Then
  \begin{align*}
    \op{ind}_{\mathcal X}f(t,z)
    &= te(z_{_{+}}z)f(z_{_{+}})\\
    \delta_{\mathbb X_{_{+}}}\!\!*\op{ind}_{\mathcal X}f(t,z)
    &= \int_{\mathbb X_{_{+}}}\op{ind}_{\mathcal X}f((1,-x_{_{+}})(t,z))\,dx_{_{+}}\\
    &= \int_{\mathbb X_{_{+}}}te(zx_{_{+}})\op{ind}_{\mathcal X}f(1,z-x_{_{+}})\,dx_{_{+}}\\
    &= \int_{\mathbb X_{_{+}}}te(zx_{_{+}} + (z_{_{+}}-x_{_{+}})z)f(z_{_{+}}-x_{_{+}})\,dx_{_{+}}\\
    &= \int_{\mathbb X_{_{+}}}te(z(z_{_{+}}-x_{_{+}}) + x_{_{+}}z)f(x_{_{+}})\,dx_{_{+}}\\
    &= te(zz_{_{+}})\int_{\mathbb X_{_{+}}}te(2x_{_{+}}z)f(x_{_{+}})\,dx_{_{+}}\\
    \delta_{\mathbb X_{_{-}}}\!\!*\delta_{\mathbb X_{_{+}}}\!\!*\op{ind}_{\mathcal X}f(t,z)
    &= \int_{\mathbb X_{_{-}}} \delta_{\mathbb X_{_{+}}}\!\!*\op{ind}_{\mathcal X}f((1,x_{_{-}})(t,z))\,dx_{_{-}}\\     
    &= \int_{\mathbb X_{_{-}}} te(zx_{_{-}})\delta_{\mathbb X_{_{+}}}\!\!*\op{ind}_{\mathcal X}f(1,z-x_{_{-}})\,dx_{_{-}}\\     
    &= \int_{\mathbb X_{_{-}}} te(zx_{_{-}} + (z-x_{_{-}})z_{_{+}})\int_{\mathbb X_{_{+}}}e(2x_{_{+}}(z-x_{_{-}}))f(x_{_{+}})\,dx_{_{+}}\,dx_{_{-}}\\     
    &= \int_{\mathbb X_{_{-}}} te(z(z_{_{-}}-x_{_{-}}) + x_{_{-}}z_{_{+}})\int_{\mathbb X_{_{+}}}e(2x_{_{+}}x_{_{-}})f(x_{_{+}})\,dx_{_{+}}\,dx_{_{-}}\\     
    &= te(zz_{_{-}})\int_{\mathbb X_{_{+}}} f(x_{_{+}})\,dx_{_{+}}\int_{\mathbb X_{_{-}}}e(2x_{_{+}}x_{_{-}}-2z_{_{+}}x_{_{-}})\,dx_{_{-}}\\     
    &= te(zz_{_{-}})\int_{\mathbb X_{_{+}}} f(x_{_{+}})\delta(x_{_{+}}-z_{_{+}})\,dx_{_{+}}.     
  \end{align*}
  In the last few steps, we used the fact that $f$ is Schwartz on $\mathbb X_{_{+}}$, and therefore the innermost integral after applying ``Fubini's theorem'' converges to a delta function as indicated, in the sense of tempered distributions on $\mathbb X_{_{+}}$.
\end{proof}

Let $f\in\mathscr M(\mathbb X_{_{+}})$ and suppose that $\mathbb X_{_{-}}$ and $\mathbb X_{_{-}}'$ are two Lagrangian complements $\mathbb X_{_{+}}$.  Let $\mathcal X=[\mathbb{X_{_{+}}X_{_{-}}}]$ and $\mathcal X'=[\mathbb{X_{_{+}}X_{_{-}}'}]$.  Then $\op{ind}_{\mathcal X}f$ is a left $\mathbb X_{_{-}}$-invariant function and $\op{ind}_{\mathcal X'}f$ is a left $\mathbb X_{_{-}}'$-invariant function.  These functions are related by
\[\op{ind}_{\mathcal X'}f(t,z) = \op{ind}_{\mathcal X}f\left(t, \frac12(\mathcal X+\mathcal X')z\right),\]
in which $\frac12(\mathcal X+\mathcal X')$ is the canonical symplectic reflection which exchanges the subspaces $\mathbb X_{_{-}}$ and $\mathbb X_{_{-}}'$ and which is the identity on $\mathbb X_{_{+}}$.  Indeed, we have $\mathcal X_{_{+}}(\mathcal X+\mathcal X')/2 = \mathcal X_{_{+}}'$, and therefore
\begin{align*}
  \op{ind}_{\mathcal X'}f(t,z)
  &= te(\mathcal X'_{_{+}}zz)f(\mathcal X_{_{+}}'z)\\
  &= te(\omega(\mathcal X_{_{+}}(\mathcal X+\mathcal X')z,(\mathcal X+\mathcal X')z)/4) f(\mathcal X_{_{+}}(\mathcal X+\mathcal X')z/2)\\
  &=\op{ind}_{\mathcal X}f(t,(\mathcal X+\mathcal X')z/2)
\end{align*}

Perhaps more interestingly, let $\mathcal X$ be a fixed real
polarization of $\mathbb T$, and let $\mathbb Y$ be any Lagrangian
subspace complementary to both eigenspaces $\mathbb X_\pm$.  The
triple convolution
\[
  \delta_{\mathbb X_{-}}\!*\delta_{\mathbb Y}*\delta_{\mathbb X_{+}}\!*f,
  \qquad f\in\mathscr S(\mathbb X_{-}\setminus\mathfrak H),
\]
is again left $\mathbb X_{-}$-invariant, and in fact is a scalar
multiple of $f$ whose value is determined by the \emph{Maslov index}.

\begin{definition}
  The \emph{Kashiwara form} on $\mathbb T\oplus\mathbb T\oplus\mathbb T$
  is the quadratic form
  \[
    \kappa(a,b,c) = \omega(a,b) + \omega(b,c) + \omega(c,a).
  \]
  The \emph{Maslov index} of an ordered triple of Lagrangian
  subspaces $(\mathbb A,\mathbb B,\mathbb C)$ is the inertial index
  (i.e., the number of positive eigenvalues minus the number of
  negative eigenvalues) of $\kappa|_{\mathbb A\oplus\mathbb B\oplus\mathbb C}$.
\end{definition}

\begin{lemma}\label{lem:maslovquadform}
  Let $(\mathbb A,\mathbb B,\mathbb C)$ be pairwise complementary
  Lagrangian subspaces, and set
  \[
    Q(c) = -\omega\!\left([\mathbb{AB}]_+ c,\,[\mathbb{AB}]_- c\right),
    \qquad c\in\mathbb C.
  \]
  Then $\tau(\mathbb A,\mathbb B,\mathbb C) = \idx(Q)$.
\end{lemma}

\begin{proof}
  Embed $\mathbb C$ into $\mathbb A\oplus\mathbb B\oplus\mathbb C$ via
  $\gamma(c) = ([\mathbb{AB}]_+ c,\, [\mathbb{AB}]_- c,\, c)$.
  Since $c = [\mathbb{AB}]_+ c + [\mathbb{AB}]_- c$, setting
  $a = [\mathbb{AB}]_+c$ and $b = [\mathbb{AB}]_-c$ gives
  \[
    \kappa(\gamma(c))
    = \omega(a,b) + \omega(b,c) + \omega(c,a)
    = -\omega(a,b)
    = -\omega([\mathbb{AB}]_+c,[\mathbb{AB}]_-c)
    = Q(c).
  \]
  (The second equality uses $c = a+b$ and the Lagrangian property.)
  On the complementary subspace
  $\mathbb A\oplus\mathbb B\oplus\{0\}\subset
   \mathbb A\oplus\mathbb B\oplus\mathbb C$, the Kashiwara form
  $\kappa(a,b,0) = \omega(a,b)$ has inertial index zero because
  the symplectic form is itself of split signature.
  The lemma follows by the additivity of the inertial index.
\end{proof}

We can now state and prove the main result of this section.

\begin{theorem}\label{thm:maslov}
  Let $\mathbb A,\mathbb B,\mathbb C$ be pairwise complementary
  Lagrangian subspaces of $\mathbb T$, with Haar measures
  $\mu_{\mathbb A},\mu_{\mathbb B},\mu_{\mathbb C}$, and let
  $\mathcal L$ denote the Liouville measure on $\mathbb T$.
  Then
  \begin{align*}
    (\mu_{\mathbb B}\otimes\mu_{\mathbb C}&\otimes\mu_{\mathbb A})
    \cdot\delta_{\mathbb B}*\delta_{\mathbb C}*\delta_{\mathbb A}
    *\op{ind}_{\mathbb{AB}}f\\
    &= \left(
        \frac{d\mu_{\mathbb A}\,d\mu_{\mathbb B}}{d\mathcal L}
       \cdot\frac{d\mu_{\mathbb B}\,d\mu_{\mathbb C}}{d\mathcal L}
       \cdot\frac{d\mu_{\mathbb C}\,d\mu_{\mathbb A}}{d\mathcal L}
      \right)^{1/2}
      e\!\left(-\frac{\sgn(h)}{4}\,\tau(\mathbb A,\mathbb B,\mathbb C)\right)
      \op{ind}_{\mathbb{AB}}f,
  \end{align*}
  where $\tau(\mathbb A,\mathbb B,\mathbb C)$ is the Maslov index of
  the triple $(\mathbb A,\mathbb B,\mathbb C)$.
\end{theorem}

The proof proceeds in three steps: first we evaluate the triple
convolution up to a scalar Gaussian integral; then we evaluate
that integral by stationary phase; finally we identify the
resulting phase with the Maslov index via
Lemma~\ref{lem:maslovquadform}.

\begin{proof}
  \textbf{Step 1: Reduction to a Gaussian integral.}
  Put $P=[\mathbb{AB}]_+$ and $Q=[\mathbb{AB}]_-=I-P$.
  Write a general element of the Heisenberg group as $z=(s,x)$
  with $s\in\mathbb R$ and $x\in\mathbb T$.  By the group law,
  right-multiplying by $b\in\mathbb B$, $c\in\mathbb C$,
  $a\in\mathbb A$ in succession gives
  \[
    z\cdot b\cdot c\cdot a
    = \Bigl(s+\tfrac12\bigl[\omega(x,b)+\omega(x+b,c)+\omega(x+b+c,a)\bigr],
      \;x+a+b+c\Bigr),
  \]
  and therefore
  \begin{align*}
    \op{ind}_{\mathbb{AB}}f(z\cdot b\cdot c\cdot a)
    &= e_h\!\Bigl(
        s+\tfrac12\bigl[\omega(x,b)+\omega(x+b,c)+\omega(x+b+c,a) \\
    &\qquad\qquad\qquad\qquad
        -\omega(P(x+a+b+c),Q(x+a+b+c))\bigr]
      \Bigr)\,
      f(P(x+a+b+c)).
  \end{align*}

  Set
  \[
    u=P(x+a+b+c)\in\mathbb A.
  \]
  Since $P$ is the projection onto $\mathbb A$ along $\mathbb B$,
  we have
  \[
    u=a+P(x+c),
    \qquad\text{hence}\qquad
    a=u-P(x+c).
  \]
  Substituting this and using that $\mathbb A$ and $\mathbb B$ are
  Lagrangian, a straightforward expansion shows that the phase
  separates as
  \[
    \op{ind}_{\mathbb{AB}}f(z\cdot b\cdot c\cdot a)
    =
    e_h\!\Bigl(
      s-\tfrac12\omega(Px,Qx)+\omega(u-Px,b)+\tfrac12\omega(Pc,c)
    \Bigr)\,f(u).
  \]
  (If one obtains $\omega(Px-u,b)$ instead, this is equivalent after
  the change of variable $b\mapsto -b$, which does not affect the Haar
  measure on $\mathbb B$.)

  We now integrate first in $b\in\mathbb B$.  By the Fourier inversion
  lemma for the complementary pair $(\mathbb A,\mathbb B)$, the kernel
  $e_h(\omega(u-Px,b))$ produces the delta distribution
  $\delta_{\mathbb A}(u-Px)$, so the $u$--integration collapses to
  $u=Px$.  Thus
  \begin{align*}
    \delta_{\mathbb B}*\delta_{\mathbb C}*\delta_{\mathbb A}
    *\op{ind}_{\mathbb{AB}}f(z)
    &=
    \frac{d\mu_{\mathbb A}\,d\mu_{\mathbb B}}{d\mathcal L}\,
    e_h\!\bigl(s-\tfrac12\omega(Px,Qx)\bigr)\,f(Px)\,
    \lambda(\mu_{\mathbb C}) \\
    &=
    \frac{d\mu_{\mathbb A}\,d\mu_{\mathbb B}}{d\mathcal L}\,
    \lambda(\mu_{\mathbb C})\,
    \op{ind}_{\mathbb{AB}}f(z),
  \end{align*}
  where
  \begin{equation}\label{eq:lambda}
    \lambda(\mu_{\mathbb C})
    =
    \lim_{\epsilon\to0^+}
    \int_{\mathbb C}
      e^{-\pi\epsilon\|c\|^2}\,
      e_h\!\bigl(\tfrac12\omega(Pc,c)\bigr)\,
      d\mu_{\mathbb C}(c).
  \end{equation}
  Since $P=[\mathbb{AB}]_+$ and $Q=I-P$, for $c\in\mathbb C$ we have
  \[
    \omega(Pc,c)=\omega(Pc,Qc)
    =-\,
    Q_{\mathbb C}(c),
  \]
  where
  \[
    Q_{\mathbb C}(c)
    =
    -\omega\!\left([\mathbb{AB}]_+c,\,[\mathbb{AB}]_-c\right)
  \]
  is the quadratic form from Lemma~\ref{lem:maslovquadform}.
  Hence
  \begin{equation}\label{eq:step1}
    \delta_{\mathbb B}*\delta_{\mathbb C}*\delta_{\mathbb A}
    *\op{ind}_{\mathbb{AB}}f
    =
    \frac{d\mu_{\mathbb A}\,d\mu_{\mathbb B}}{d\mathcal L}\,
    \lambda(\mu_{\mathbb C})\,
    \op{ind}_{\mathbb{AB}}f.
  \end{equation}

  \textbf{Step 2: Evaluation of the Gaussian integral.}
  We use the standard Fresnel formula: if $R$ is a non-degenerate real
  quadratic form on a real vector space $\mathbb V$ with Haar measure
  $\mu$, then
  \begin{equation}\label{eq:gaussianphase}
    \lim_{\epsilon\to0^+}
    \int_{\mathbb V}
      e^{-\pi\epsilon\|v\|^2}\,
      e_h\!\bigl(\tfrac12 R(v)\bigr)\,
      d\mu(v)
    =
    |\det R|^{-1/2}\,
    e\!\left(\frac{\idx(hR)}{4}\right),
  \end{equation}
  where $\det R$ is computed relative to $\mu$ and
  $\idx(hR)=\sgn(h)\,\idx(R)$ is the inertial index.
  Indeed, after diagonalising $R$ one reduces to the one-variable
  integral
  \[
    \int_{-\infty}^{\infty}
      e^{-\pi\epsilon x^2}\,e_h\!\bigl(\tfrac12 qx^2\bigr)\,dx
    =
    \int_{-\infty}^{\infty}
      e^{-\pi(\epsilon-i h q)x^2}\,dx
    =
    (\epsilon-i h q)^{-1/2},
  \]
  whose limit has phase $e(\sgn(hq)/4)$.

  Applying \eqref{eq:gaussianphase} to \eqref{eq:lambda} with
  \[
    R_{\mathbb C}(c)=\omega(Pc,c)=-Q_{\mathbb C}(c),
  \]
  we obtain
  \[
    \idx(hR_{\mathbb C})
    =
    \idx(-hQ_{\mathbb C})
    =
    -\sgn(h)\,\idx(Q_{\mathbb C})
    =
    -\sgn(h)\,\tau(\mathbb A,\mathbb B,\mathbb C),
  \]
  by Lemma~\ref{lem:maslovquadform}.  Therefore
  \begin{equation}\label{eq:lambdaeval}
    \lambda(\mu_{\mathbb C})
    =
    |\det Q_{\mathbb C}|^{-1/2}\,
    e\!\left(
      -\frac{\sgn(h)}{4}\,\tau(\mathbb A,\mathbb B,\mathbb C)
    \right).
  \end{equation}

  \textbf{Step 3: Identifying the determinant prefactor.}
  Fix an identification $\mathbb T\cong\mathbb K\oplus\mathbb K$, with
  $\mathbb K$ an $n$-dimensional Euclidean space, such that
  \[
    \omega((k_1,k_2),(k_1',k_2'))
    =
    k_1\!\cdot\!k_2' - k_2\!\cdot\!k_1'.
  \]
  Write the three pairwise complementary Lagrangians as graphs of
  symmetric endomorphisms:
  \[
    \mathbb A=\{(k,Ak)\},\qquad
    \mathbb B=\{(k,Bk)\},\qquad
    \mathbb C=\{(k,Ck)\}.
  \]
  For $c=(k,Ck)\in\mathbb C$, write
  \[
    c=a+b,\qquad a\in\mathbb A,\ b\in\mathbb B.
  \]
  Solving for $a$ gives
  \[
    a=\bigl((A-B)^{-1}(C-B)k,\;
      A(A-B)^{-1}(C-B)k\bigr),
  \]
  hence
  \begin{align*}
    Q_{\mathbb C}(k,Ck)
    &=
    -\omega(a,c) \\
    &=
    -\,k^T(C-B)(A-B)^{-1}(C-A)\,k.
  \end{align*}
  Consequently,
  \[
    |\det Q_{\mathbb C}|
    =
    \frac{|\det(C-B)|\,|\det(C-A)|}{|\det(A-B)|}.
  \]
  On the other hand, relative to the Lebesgue measures on the graph
  parameters, the product Haar measures satisfy
  \[
    \frac{d\mu_{\mathbb A}\,d\mu_{\mathbb B}}{d\mathcal L}
    = |\det(A-B)|^{-1},\qquad
    \frac{d\mu_{\mathbb B}\,d\mu_{\mathbb C}}{d\mathcal L}
    = |\det(B-C)|^{-1},\qquad
    \frac{d\mu_{\mathbb C}\,d\mu_{\mathbb A}}{d\mathcal L}
    = |\det(C-A)|^{-1}.
  \]
  Therefore
  \begin{align*}
    \frac{d\mu_{\mathbb A}\,d\mu_{\mathbb B}}{d\mathcal L}\,
    |\det Q_{\mathbb C}|^{-1/2}
    &=
    |\det(A-B)|^{-1}
    \left(
      \frac{|\det(A-B)|}
           {|\det(C-B)|\,|\det(C-A)|}
    \right)^{1/2} \\
    &=
    \bigl(
      |\det(A-B)|\,|\det(B-C)|\,|\det(C-A)|
    \bigr)^{-1/2} \\
    &=
    \left(
      \frac{d\mu_{\mathbb A}\,d\mu_{\mathbb B}}{d\mathcal L}\,
      \frac{d\mu_{\mathbb B}\,d\mu_{\mathbb C}}{d\mathcal L}\,
      \frac{d\mu_{\mathbb C}\,d\mu_{\mathbb A}}{d\mathcal L}
    \right)^{1/2},
  \end{align*}
  since $|\det(B-C)|=|\det(C-B)|$.

  Substituting \eqref{eq:lambdaeval} into \eqref{eq:step1} and using
  the preceding identity yields
  \[
    (\mu_{\mathbb B}\otimes\mu_{\mathbb C}\otimes\mu_{\mathbb A})
    \cdot\delta_{\mathbb B}*\delta_{\mathbb C}*\delta_{\mathbb A}
    *\op{ind}_{\mathbb{AB}}f
    =
    \left(
      \frac{d\mu_{\mathbb A}\,d\mu_{\mathbb B}}{d\mathcal L}\,
      \frac{d\mu_{\mathbb B}\,d\mu_{\mathbb C}}{d\mathcal L}\,
      \frac{d\mu_{\mathbb C}\,d\mu_{\mathbb A}}{d\mathcal L}
    \right)^{1/2}
    e\!\left(
      -\frac{\sgn(h)}{4}\,\tau(\mathbb A,\mathbb B,\mathbb C)
    \right)
    \op{ind}_{\mathbb{AB}}f.
  \]
  This is the desired formula.
\end{proof}

\section{Bargmann and theta transforms}
\label{sec:bargmann}
Now let $\mathcal J$ now be an imaginary polarization.  For $t\in U(1)$ and $z\in\mathbb T$, put $\eta_{\mathcal J}(t,z) = t e(\mathcal J_{_{+}}zz)$.  Then:

\begin{lemma}\label{lem:eta_holomorphic}
  $\eta_{\mathcal J}$ is left $\mathcal J$-holomorphic and right $\mathcal J$-anti\-holomorphic.
\end{lemma}
\begin{proof}
  \begin{align*}
    L(1,w_{_{-}})\eta(t,z) &= te(w_{_{-}}z)\eta(1,z+w_{_{-}}) = te(w_{_{-}}z_{_{+}} + z_{_{+}}(z+w_{_{-}})) = te(z_{_{+}}z) =  \eta(t,z).\\
    R(1,w_{_{+}})\eta(t,z) &= te(zw_{_{+}})\eta(1,z+w_{_{+}}) = te(z_{_{-}}w_{_{+}} + (z_{_{+}}+w_{_{+}})z) \\
                      &= te(z_{_{+}}z + z_{_{-}}w_{_{+}} + w_{_{+}}z_{_{-}}) = te(z_{_{+}}z) =  \eta(t,z).
  \end{align*}
\end{proof}

\begin{proof}[Proof of theorem]
  Left $\mathcal J$-holomorphicity of $\eta_{\mathcal J}*f$ follows from
  Lemma~\ref{lem:eta_holomorphic}: $\eta_{\mathcal J}$ is left $\mathcal J$-holomorphic,
  hence so is every left-translate, and the convolution inherits this invariance.

  It remains to show that $\eta_{\mathcal J}*f$ is Schwartz.
  The calculation in the paper already gives
  \begin{equation}\label{eq:bargmann_conv}
    \eta_{\mathcal J}*f(t,z)
    = t\, e(\mathcal J_{_{+}}zz)
      \int_{\mathbb T} e(y_{_{+}}y - 2z_{_{+}}y)\,f(1,y)\,dy,
  \end{equation}
  where $z_{_{+}} = \mathcal J_{_{+}}z$ and $y_{_{+}} = \mathcal J_{_{+}}y$.
  Since all information about the $t$-dependence is the overall factor of $t\in U(1)$,
  it suffices to study the function $F:\mathbb T\to\mathbb C$ defined by
  \[F(z) = e(\mathcal J_{_{+}}zz)\,G(z_{_{+}}),\qquad
    G(w) = \int_{\mathbb T} e(y_{_{+}}y)\,e(-2wy)\,f(1,y)\,dy.\]

  \medskip\noindent\textbf{Positivity and the Gaussian factor.}
  By the positivity assumption $i\omega(x,\mathcal Jx)>0$ for all $x\ne 0$,
  the quadratic form $x\mapsto i\cdot 2\mathcal J_{_{+}}xx$ is real and positive-definite on $\mathbb T$.
  Since $e(z) = e^{\pi iz}$ and $\mathcal J_{_{+}}zz$ is complex-valued,
  \[|e(\mathcal J_{_{+}}zz)| = e^{-\pi\operatorname{im}(\mathcal J_{_{+}}zz)}.\]
  Positive-definiteness gives a constant $c>0$ such that
  $\operatorname{im}(\mathcal J_{_{+}}zz) \ge c\|z\|^2$ for all $z\in\mathbb T$, so
  \begin{equation}\label{eq:gaussian_decay}
    |e(\mathcal J_{_{+}}zz)| \le e^{-\pi c\|z\|^2}.
  \end{equation}

  \medskip\noindent\textbf{$G$ is Schwartz in $z_{_{+}}$.}
  The function $y\mapsto e(y_{_{+}}y)f(1,y)$ is in $\mathscr S(\mathbb T)$:
  $e(y_{_{+}}y) = e(\mathcal J_{_{+}}yy)$ has modulus $e^{-\pi\operatorname{im}(\mathcal J_{_{+}}yy)}\le e^{-\pi c\|y\|^2}$
  by~\eqref{eq:gaussian_decay} applied with $z=y$, so multiplication
  by $e(y_{_{+}}y)$ maps $\mathscr S(\mathbb T)$ to $\mathscr S(\mathbb T)$.
  Therefore $G(w)$ is (up to the identification of $\mathbb X_{_{+}}$ with the dual of $\mathbb T$
  via $\omega$) the Fourier transform of a Schwartz function, and is itself Schwartz:
  for every polynomial $p(w)$ and every $N\ge 0$,
  \begin{equation}\label{eq:G_schwartz}
    |p(w)\,G(w)| \le C_{p,N}(1+\|w\|)^{-N}.
  \end{equation}
  Concretely, differentiating under the integral sign,
  $\partial^{\alpha}_w G(w) = \int_{\mathbb T}(-2\pi i y)^\alpha\,
   e(y_{_{+}}y-2wy)f(1,y)\,dy$,
  and each such integral is bounded uniformly in $w$ because $y\mapsto y^\alpha e(y_{_{+}}y)f(1,y)$
  is in $L^1(\mathbb T)$; the rapid decay in $w$ then follows from repeated integration
  by parts in $y$.

  \medskip\noindent\textbf{Schwartz estimates for $F$.}
  We must estimate $|z^\beta\,\partial^\alpha_z F(z)|$ for all multi-indices $\alpha,\beta$.
  Write $z = z_{_{+}}+z_{_{-}}$ with $z_\pm = \mathcal J_\pm z$.

  \emph{Derivatives in $z_{_{-}}$.}
  Left $\mathcal J$-holomorphicity of $F$ (as a function on $\mathbb T$)
  means $\partial_{z_{_{-}}}F = 0$ in the sense of the Cauchy--Riemann equation
  $\partial_{_{-}}F = i\pi F\,\omega(z_{_{+}},dz_{_{-}})$.
  Consequently, $\partial_{z_{_{-}}}^k F$ can be expressed as a
  sum of products of $F$ with polynomials in $z_{_{+}}$, so all
  $z_{_{-}}$-derivatives reduce to $z_{_{+}}$-derivatives times polynomial
  factors; it suffices to bound $\partial_{z_{_{+}}}^\alpha F$ and $z^\beta F$.

  \emph{$z_{_{+}}$-derivatives.}
  Differentiating \eqref{eq:bargmann_conv} in $z_{_{+}}$:
  \[\partial_{z_{_{+}}}^\alpha F(z)
    = \sum_{\alpha'\le\alpha}\binom{\alpha}{\alpha'}
      (\partial^{\alpha'}\,e(\mathcal J_{_{+}}zz))\cdot(\partial^{\alpha-\alpha'} G)(z_{_{+}}).\]
  Each $z_{_{+}}$-derivative of $e(\mathcal J_{_{+}}zz)$ produces a polynomial factor in $z_{_{+}}$
  times $e(\mathcal J_{_{+}}zz)$, so the $e(\mathcal J_{_{+}}zz)$-factor is never differentiated away.
  Thus $$|\partial_{z_{_{+}}}^\alpha F(z)| \le p_\alpha(z_{_{+}})\,e^{-\pi c\|z\|^2}\,\sup_w|(\partial^{\alpha-\alpha'}G)(w)|$$
  for some polynomial $p_\alpha$, and the right-hand side is rapidly decreasing in $\|z\|$.

  \emph{Polynomial growth.}
  For any monomial $z^\beta = z_{_{+}}^{\beta_{_{+}}}z_{_{-}}^{\beta_{_{-}}}$,
  the factor $|z^\beta|$ grows at most polynomially in $\|z\|$, while
  $e^{-\pi c\|z\|^2}$ decays faster than any power, so the product
  $|z^\beta\,\partial^\alpha F(z)|$ is bounded.  Indeed, for every $N$ and $\alpha,\beta$,
  \[|z^\beta\,\partial^\alpha_z F(z)|
    \le C_{\alpha,\beta,N}\,(1+\|z\|)^{-N}\]
  by combining the Gaussian bound~\eqref{eq:gaussian_decay} with
  the Schwartz estimates~\eqref{eq:G_schwartz} for $G$.
  This establishes that $F$, and hence $\eta_{\mathcal J}*f$, is Schwartz.
\end{proof}

  \subsection{Theta transform}
  Let $\Lambda$ be a self-dual lattice in $\mathbb T$, meaning that the symplectic form is integral on $\Lambda$, and its dual lattice is itself under the symplectic form.  (The latter condition is equivalent to the statement that $\mathbb T/\Lambda$ has Liouville measure 1.)  Consider a subgroup $\sigma(\Lambda)$ of $\mathfrak H$ that splits the projection $\pi$ onto $\Lambda$.  We construct all such splittings.  Let $Q$ be an $F_2$-valued quadratic form on $\Lambda/2\Lambda$, such that $Q(x+y)+Q(x)+Q(y)+\omega(x,y)\equiv 0\pmod 2$.  Then let $\sigma_Q(\lambda) = (Q(\lambda),\lambda)$.  It is easily shown that $\sigma_Q(\Lambda)$ is an (abelian) subgroup of $\mathfrak H$, and $\sigma_Q$ is clearly a splitting.  We denote by $\Lambda_Q$ the subgroup $\sigma_Q(\Lambda)$ determined by the data $Q$.

  Let $\delta_{\Lambda_Q}$ be the counting measure supported on the discrete subgroup $\Lambda_Q$ of $\mathfrak H$.  Then $\delta_{\Lambda_Q}$ is a tempered distribution on $\mathfrak H$ which is $\Lambda_Q$-biinvariant.
  \begin{definition}
    Given a discrete abelian subgroup $\Lambda_Q$ lifing a lattice $\Lambda$ as above, the theta transform of $f\in\mathscr S(\mathfrak H)$ is the function
    \[f * \delta_{\Lambda_Q}.\]
  \end{definition}
  \begin{theorem}\label{thm:theta_transform}
    $f * \delta_{\Lambda_Q}\in \mathscr S(\mathfrak H/\Lambda_Q)$ for all $f\in\mathscr S(\mathfrak H)$.
  \end{theorem}
  \begin{proof}
    $f*\delta$ is obviously right $\Lambda_Q$-invariant and smooth.  But $\mathfrak H/\Lambda_Q$ is a (compact) torus, so all smooth functions are Schwartz.
  \end{proof}

  Let $\mathcal J$ be a positive imaginary polarization.  For any $f\in\mathscr S(\mathbb J_{_{-}}\setminus\mathfrak H)$, the convolution $f*\delta_{\Lambda_Q}$ is left $\mathbb J_{_{-}}$-invariant, i.e., left $\mathcal J$-holomorphic.  It is also right $\Lambda_Q$-invariant.  Therefore it also has a quasiperiodicity with respect to the {\em left} action of $\Lambda_Q$:
  \begin{theorem}\label{thm:theta_quasiperiodicity}
    For $w\in \Lambda_Q$,
    \[L(w)(f*\delta_{\Lambda_Q})(z) = e(2\,zw)\,f*\delta_{\Lambda_Q}(z).\]
  \end{theorem}

  \begin{proposition}\label{prop:bargmann_schrodinger}
    Let $\mathbb Y=\mathbb H(s)$ and choose a complementary Lagrangian $\mathbb X$, so that
    \[
      \mathbb H(u)=\{H(u)\eta+\eta\mid \eta\in\mathbb Y\},
    \]
    where $H(u):\mathbb Y\to\mathbb X$ is symmetric.  Let $\hat\phi_0\in\mathscr S(\mathbb Y)$, and let
    \[
      \phi_u(v,x)=\int_{\mathbb Y} e_h\!\left(v+x^T\eta+\tfrac12\eta^TH(u)\eta\right)\hat\phi_0(\eta)\,\op{d}_h\eta
    \]
    be the corresponding Schr\"odinger evolution in the form of Theorem~\ref{AuTheorem}.  Then for every positive imaginary polarization $\mathcal J$, the Bargmann transform of $\phi_u$ is
    \[
      (\eta_{\mathcal J}*\phi_u)(1,z)
      =
      e(\mathcal J_{_{+}}zz)
      \int_{\mathbb Y}
      e\!\left(\eta_{_{+}}\eta-2z_{_{+}}\eta\right)
      e_h\!\left(\tfrac12\eta^TH(u)\eta\right)\hat\phi_0(\eta)\,\op{d}_h\eta.
    \]
    In particular, relative to the Bargmann kernel at time $s$, Schr\"odinger evolution acts by multiplication by the same quadratic phase
    \[
      e_h\!\left(\tfrac12\eta^T(H(u)-H(s))\eta\right)
    \]
    that appears in the real-polarized model.  Thus the same Hamiltonian curve $H(u)$ governs both the Schr\"odinger and the Bargmann realizations.
  \end{proposition}
  \begin{proof}
    Apply the convolution formula \eqref{eq:bargmann_conv} to the function $\phi_u$.  The only $u$-dependence of $\phi_u$ is the quadratic factor
    $e_h(\tfrac12\eta^TH(u)\eta)$ appearing in its oscillatory representation, so substituting that representation into \eqref{eq:bargmann_conv} yields exactly the displayed formula.  Replacing $u$ by $s$ and comparing the two expressions shows that the passage from time $s$ to time $u$ is effected by multiplication by
    $e_h(\tfrac12\eta^T(H(u)-H(s))\eta)$ inside the Bargmann integral.
  \end{proof}

  \begin{corollary}\label{cor:theta_schrodinger}
    In the arithmetic situation of Theorems~\ref{thm:theta_transform} and~\ref{thm:theta_quasiperiodicity}, the theta transform of the evolved Bargmann datum is obtained by summing the kernel of Proposition~\ref{prop:bargmann_schrodinger} over $\Lambda_Q$.  Consequently Schr\"odinger evolution preserves the theta-function automorphy law and acts on theta data by the same quadratic phase factor.
  \end{corollary}
  \begin{proof}
    The theta transform is convolution with the lattice distribution $\delta_{\Lambda_Q}$.  Since Proposition~\ref{prop:bargmann_schrodinger} identifies Schr\"odinger evolution in the Bargmann model with multiplication by the quadratic phase inside the holomorphic kernel, summing over $\Lambda_Q$ preserves the quasiperiodicity of Theorem~\ref{thm:theta_quasiperiodicity} and yields the stated description of the evolved theta datum.
  \end{proof}

  \section{Schr\"odinger evolution in a plane wave}
  \label{sec:schrodinger}
Consider a Rosen plane wave on $\mathbb U\times\mathbb R\times\mathbb X$.  Fix $u=0$ in $\mathbb U$.  Let $\mathbb H(u)$ be the Hamiltonian curve.  A {\em field bundle} is a bundle $\mathfrak F$ over $\mathbb M$, equipped with an action of the Heisenberg symmetry group, and having a connection on the central null geodesic, with parallel transport map $\Gamma_{st}:\mathfrak F_s\to\mathfrak F_t$.  An example of a field bundle is the bundle $\mathcal O_{\mathbb U}(w)$ of functions of projective degree $w$, pulled back to $\mathbb M$ so that the Heisenberg action is the trivial action.  Then multiplication by $(|g(t)|/|g(s)|)^{w/2n}$ maps $\mathcal O_{\mathbb U}(w)_{u=s}\to\mathcal O_{\mathbb U}(w)_{u=t}$, and so the natural connection is $\Gamma_{st} = (|g(t)|/|g(s)|)^{w/2n}$.  In general, natural connections will be more complicated.

We now describe the general Schr\"odinger evolution in a plane wave, beginning from initial data at $u=s$.  Let $\mathbb Y=\mathbb H(s)$, and the Hamiltonian curve is given by $\mathbb H(u) =\{ H(u)y + y |y\in\mathbb Y\}$ where $H(u):\mathbb Y\to\mathbb X$ is a symmetric transformation with $\dot H$ positive definite.
% \begin{equation}\label{BigSchrodinger}
% \begin{array}{cl}
%   &\mathscr S(\mathbb X)\otimes\mathfrak F(s)\xrightarrow{\mathcal F_{\mathbb X\mathbb H(s)}}\mathscr S^1(\mathbb H(s))\otimes\mathfrak F(s) \xrightarrow{\text{ind}}
%   \mathscr H^1(\mathbb X^\circ)\otimes\mathfrak F(s)\\
%   &\qquad\xrightarrow{\Gamma_{su}}\mathscr H^1(\mathbb X^\circ)\otimes\mathfrak F(u)\xrightarrow{\mathcal W_{\mathbb H(u)\mathbb X}}\mathscr H(\mathbb H(u)^\circ)\otimes\mathfrak F(u)\xrightarrow{\text{res}}\mathscr S(\mathbb X)\otimes\mathfrak F(u)
% \end{array}
% \end{equation}
\begin{equation}\label{BigSchrodinger}
\begin{array}{cl}
  &\mathscr S(\mathbb X)\otimes\mathfrak F(s)\xrightarrow{\ind}\mathscr S(\mathbb H(s)\setminus \mathfrak H)\otimes\mathfrak F(s) \xrightarrow{\delta_{\mathbb X}*}
  \mathscr S(\mathbb X\setminus\mathfrak H)\otimes\mathfrak F(s)\\
  &\qquad\xrightarrow{\Gamma_{su}}\mathscr S(\mathbb X\setminus\mathfrak H)\otimes\mathfrak F(u)\xrightarrow{\delta_{\mathbb H(u)}*}\mathscr S(\mathbb H(u)\setminus\mathfrak H)\otimes\mathfrak F(u)\xrightarrow{\text{res}}\mathscr S(\mathbb X)\otimes\mathfrak F(u)
\end{array}
\end{equation}
More briefly, this collapses into a map
$$\Phi_{\mathbb X,\mathfrak F}(s,u) : \mathscr S(\mathbb X)\otimes\mathfrak F(s) \to \mathscr S(\mathbb X)\otimes\mathfrak F(u) $$
which we shall call the {\em Schr\"odinger evolution}.

We show the sense in which $\Phi(s,u)$ yields a solution of the Schr\"odinger equation, in the case of $\mathfrak F = \mathcal O_{\mathbb U}(-n/2)$.  Let $\mathbb Y=\mathbb H(0)$, and use these to define a splitting of $\mathbb T$ as $\mathbb X\oplus\mathbb Y$.  Let $H(u):\mathbb Y\to\mathbb X$ be symmetric, such that $\mathbb H(u)=\{H(u)y + y|y\in\mathbb Y\}$.  Note that $H(0)=0$ and $\dot H(u)$ is positive-definite.

Let $\phi_0\in\mathscr S^\alpha(\mathbb X)$ denote the initial data for the evolution (we will deal with the degrees later), and $\hat \phi_0\in\mathscr S^{1-\alpha}(\mathbb Y)$ the Fourier transform.  The induced function $\psi$ in $\mathscr H^{1-\alpha}(\mathbb X^\circ)$ is
$$\psi(v,x,y) = e_h(v+2^{-1}x^Ty)\hat\phi_0(y).$$
We have
\begin{align*}
  \psi((v,x,y)\cdot (0,x',0))
  &= \psi(v-2^{-1}y^Tx', x+x',y) \\
  &= e_h(v-2^{-1}y^Tx'+2^{-1}(x+x')^Ty)\hat\phi_0(y)\\
  &= e_h(v+2^{-1}x^Ty)\hat\phi_0(y)\\
  &=\psi(t,x,y),
\end{align*}
so that $\psi\in\mathscr H^{1-\alpha}(\mathbb X^\circ)$.

Next, right convolution by the Lagrangian distribution $\delta_{\mathbb H(u)}$ (i.e.\ averaging over the subgroup $\mathbb H(u)$) is
\begin{align*}
  (\delta_{\mathbb H(u)}*\psi)(v,x,y)
  &= \int_{\mathbb H(u)} \psi((v,x,y)\cdot \xi)\,\op{d}_h\!\xi\\
  &= \int_{\mathbb Y} \psi((v,x,y)\cdot (0,H(u)\eta,\eta)\,\op{d}_h\!\eta\\  
  &= \int_{\mathbb Y} \psi(v + 2^{-1}(x^T\eta - y^TH\eta),x+H\eta,y+\eta)\,\op{d}_h\!\eta\\
  &= \int_{\mathbb Y} e_h(v+ 2^{-1}(-y^TH\eta+x^T\eta  +(x+H\eta)^T(y+\eta)))  \hat\phi_0(y+\eta)\,\op{d}_h\!\eta\\
\end{align*}
% Next, the Weyl transform of $\psi$ is
% \begin{align*}
%   \mathcal W_{\mathbb H(u)\mathbb X}\psi(v,x,y)
%   &= \int_{\mathbb H(u)} \psi((v,x,y)\cdot \xi)\,\op{d}_h\!\xi\\
%   &= \int_{\mathbb Y} \psi((v,x,y)\cdot (0,H(u)\eta,\eta)\,\op{d}_h\!\eta\\  
%   &= \int_{\mathbb Y} \psi(v + 2^{-1}(x^T\eta - y^TH\eta),x+H\eta,y+\eta)\,\op{d}_h\!\eta\\
%   &= \int_{\mathbb Y} e_h(v+ 2^{-1}(-y^TH\eta+x^T\eta  +(x+H\eta)^T(y+\eta)))  \hat\phi_0(y+\eta)\,\op{d}_h\!\eta\\
% \end{align*}
To simplify this integral further, we would normally make the change of variables $\bar\eta=y+\eta$, but since this is immediately to be restricted to $\mathbb X$, we simply take $y=0$:
\[\phi(u,v,x) = (\delta_{\mathbb H(u)}*\psi)(v,x,0) = \int_{\mathbb Y} e_h(v + x^T\eta + 2^{-1}\eta^TH(u)\eta)  \hat\phi_0(\eta)\,\op{d}_h\!\eta.\]
%\[\phi(u,v,x) = \mathcal W_{\mathbb H(u)\mathbb X}\psi(v,x,0) = \int_{\mathbb Y} e_h(v + x^T\eta + 2^{-1}\eta^TH(u)\eta)  \hat\phi_0(\eta)\,\op{d}_h\!\eta.\]
So $\phi\in\mathscr H^\alpha(\mathbb H(u)^\circ)$ now clearly solves a Schr\"odinger equation of the form
\[ [4\pi i h\partial_u  -  \dot H(u)(\partial_x,\partial_x)]\phi(u,v,x) = 0.\]

%%%%
\begin{theorem}[Local intertwiner]\label{thm:schrodinger-evolution}
  Suppose that $\mathbb H(u)$ is a positive Lagrangian curve, and that
  $\mathbb X$ and $\mathbb X'$ are Lagrangian subspaces of $\mathbb T$,
  each complementary to $\mathbb H(u)$ for all $u\in\mathbb U$.
  Fix a base point $s\in\mathbb U$.  In the splitting
  $\mathbb T=\mathbb H(s)\oplus\mathbb X$, write
  $\mathbb X'=\{Yz_1+z_1\mid z_1\in\mathbb X\}$ ($Y:\mathbb X\to\mathbb H(s)$
  symmetric) and $\mathbb H(u)=\{Kv+v\mid v\in\mathbb H(s)\}$
  ($K:\mathbb H(s)\to\mathbb X$ symmetric, positive-definite).
  Assume that $\mathbb X$ and $\mathbb X'$ are \emph{sufficiently close
  relative to the curve}, meaning that the operator
  $I-YK:\mathbb H(s)\to\mathbb H(s)$ is positive-definite for all
  $u\in\mathbb U$.

  Fix a field bundle $\mathfrak F$ with parallel transport
  $\Gamma$, and let
  \[\Phi_{\mathbb X}(s,u) : \mathscr S(\mathbb X)\otimes\mathfrak F(s) \to \mathscr S(\mathbb X)\otimes\mathfrak F(u)\]
  be the Schr\"odinger evolution~\eqref{BigSchrodinger} in the
  $\mathbb X$-picture.  Then the diagram
  \[\begin{tikzcd}
      \mathscr S(\mathbb X)\otimes\mathfrak F(s) \arrow[r,"\Phi_{\mathbb X}{(s,u)}"]\arrow[dd,"{\rho_s}"]& \mathscr S(\mathbb X)\otimes\mathfrak F(u)\arrow[dd,swap,"{\rho_u}"]\\
      &\\
      \mathscr S(\mathbb X')\otimes\mathfrak F(s) \arrow[r,swap,"\Phi_{\mathbb X'}{(s,u)}"]& \mathscr S(\mathbb X')\otimes\mathfrak F(u)\\
    \end{tikzcd}\]
  commutes, where
  \begin{align}\label{eq:rho-forward}
  \rho_sf(z) &= e_h\left(\tfrac12\omega([\mathbb H(s)\mathbb X]_+z,z)\right)f\left(-\monib{\mathbb{XX}'}{\mathbb H(s)}z\right)
  \end{align}
  and $\rho_u$ is given by the same expression with $u$ replacing $s$:
  \begin{align}\label{eq:rho-forward-u}
    \rho_uf(z) &= e_h\left(\tfrac12\omega([\mathbb H(u)\mathbb X]_+z,z)\right)f\left(-\monib{\mathbb{XX}'}{\mathbb H(u)}z\right).
  \end{align}

  Furthermore, the inverse is
  \begin{align}\label{eq:rho-inverse}
  \rho_s^{-1}f(z) &= e_h\left(\tfrac12\omega([\mathbb H(s)\mathbb X']_+z,z)\right)f\left(-\monib{\mathbb{X'X}}{\mathbb H(s)}z\right)
  \end{align}
  (and likewise for $\rho_u^{-1}$).
\end{theorem}
\begin{proof}
  We expand the definition of $\Phi$ from \eqref{BigSchrodinger}
  and show that each square of the resulting diagram
    \[\begin{tikzcd}
        \mathscr S(\mathbb X)_s\arrow[r,"\mathcal F_{\mathbb X\mathbb H(s)}"]\arrow[dd,"{\rho_s}"]
        & \mathscr S^1(\mathbb H(s))_s\arrow[r,"\text{ind}"]
        & \mathscr H^1(\mathbb X^\circ)_s\arrow[r,"\delta_{\mathbb H(u)}*\Gamma_{su}"]
        & \mathscr H(\mathbb H(u)^\circ)_u\arrow[r,"\text{res}"]\arrow[dd,equal]
        & \mathscr S(\mathbb X)_u\arrow[dd,swap,"{\rho_u}"]\\
        &&&&\\
        \mathscr S(\mathbb X')_s\arrow[r,"\mathcal F_{\mathbb X'\mathbb H(s)}"]
        & \mathscr S^1(\mathbb H(s))_s \arrow[r,"\text{ind}"]
        & \mathscr H^1(\mathbb X'^\circ)_s\arrow[r,"\delta_{\mathbb H(u)}*\Gamma_{su}"]
        & \mathscr H(\mathbb H(u)^\circ)_u\arrow[r,"\text{res}"]
        & \mathscr S(\mathbb X')_u.
      \end{tikzcd}\]
  commutes.

  \medskip\noindent\textbf{Left square.}\enspace
    Going across the top:
    \begin{align*}
      \mathcal Ff(\xi)
      &= \int_{\mathbb X}e_h(\omega(z,\xi))\,f\left(z\right)\,\op{d}_h\!z,\\
      \op{ind}(\mathcal Ff)(t,\xi)
      &= \int_{\mathbb X}e_h\left(t-\tfrac12\omega(P_+\xi,\xi)+\omega(z,P_+\xi)\right)\,f\left(z\right)\,\op{d}_h\!z,
    \end{align*}
    where we abbreviate $P_+=[\mathbb H(s)\mathbb X]_+$ and,
    below, $P'_+=[\mathbb H(s)\mathbb X']_+$.
    Right convolution by $\delta_{\mathbb H(u)}$ gives
    \begin{align}
      (\delta_{\mathbb H(u)}&*\op{ind}(\mathcal Ff))(t,\xi) \notag\\
      &= \int_{\mathbb H(u)}\int_{\mathbb X}e_h\!\left(t+\tfrac12\omega(\xi,\mu)-\tfrac12\omega(P_+(\xi\!+\!\mu),\xi\!+\!\mu)+\omega(z,P_+(\xi\!+\!\mu))\right)f(z)\,\op{d}_h\!z\,\op{d}_h\!\mu. \label{eq:top-integral}
    \end{align}

    Going across the bottom, after the change of variables
    $z\mapsto -\monib{\mathbb{XX}'}{\mathbb H(s)}z$ that replaces the
    $\mathbb X'$-integral by an $\mathbb X$-integral
    (cf.~\eqref{eq:rho-forward}):
    \begin{align}
      (\delta_{\mathbb H(u)}&*\op{ind}(\mathcal F(\rho_sf)))(t,\xi) \notag\\
      &= \int_{\mathbb H(u)}\int_{\mathbb X}e_h\!\left(t+\tfrac12\omega(\xi,\mu)-\tfrac12\omega(P'_+(\xi\!+\!\mu),\xi\!+\!\mu) - \tfrac12\omega(P'_+z,z)+\omega(z,P'_+(\xi\!+\!\mu))\right) \notag\\
      &\qquad\qquad\qquad f(z)\,\op{d}_h\!z\,\op{d}_h\!\mu. \label{eq:bot-integral}
    \end{align}
    It suffices to show that for each fixed $z\in\mathbb X$, the
    inner integrals over $\mu\in\mathbb H(u)$ in
    \eqref{eq:top-integral} and \eqref{eq:bot-integral} are equal.

    Select coordinates relative to the splitting
    $\mathbb T=\mathbb H(s)\oplus\mathbb X$:
    \[
      P_+ = \begin{bmatrix}I&0\\0&0\end{bmatrix},\quad
      P'_+ = \begin{bmatrix}I&-Y\\0&0\end{bmatrix},\quad
      \mathbb H(u) = \{ [\mu_0, K\mu_0] \mid \mu_0\in \mathbb H(s)\}.
    \]
    Write $\xi=[\xi_0,\xi_1]$, $\mu=[\mu_0,K\mu_0]$, $z=[0,z_1]$,
    and set $w=\xi_1+z_1$.

    \smallskip\noindent\emph{Top exponent.}\enspace
    Expanding the exponent of~\eqref{eq:top-integral}
    (excluding the overall $t$):
    \begin{align}
      E_{\mathrm{top}}
      &= \tfrac12\omega(\xi,\mu)
         -\tfrac12\omega(P_+(\xi\!+\!\mu),\xi\!+\!\mu)
         + \omega(z,P_+(\xi\!+\!\mu)) \notag\\
      &= -\tfrac12\mu_0 K\mu_0 - w\,\mu_0
         -\tfrac12\xi_1\xi_0 - z_1\xi_0. \label{eq:Etop}
    \end{align}
    (Here $\mu_0 K\mu_0$ means $\mu_0^TK\mu_0$, etc.)
    Completing the square in $\mu_0$ (using $K>0$):
    \begin{equation}\label{eq:top-cs}
      E_{\mathrm{top}}
      = -\tfrac12\bigl(\mu_0+K^{-1}w\bigr)K
         \bigl(\mu_0+K^{-1}w\bigr)
        + \tfrac12 wK^{-1}w
        - \tfrac12\xi_1\xi_0 - z_1\xi_0.
    \end{equation}

    \smallskip\noindent\emph{Bottom exponent.}\enspace
    The same expansion applied to~\eqref{eq:bot-integral},
    now using $P'_+=\bigl[\begin{smallmatrix}I&-Y\\0&0\end{smallmatrix}\bigr]$,
    gives (after a direct computation):
    \begin{equation}\label{eq:Ebot}
      E_{\mathrm{bot}}
      = -\tfrac12\mu_0\,Q\,\mu_0 - w(I\!-\!YK)\mu_0
        + \tfrac12 wYw
        - \tfrac12\xi_1\xi_0 - z_1\xi_0,
    \end{equation}
    where $Q = K(I\!-\!YK) = K - KYK$.
    (The off-diagonal block $-Y$ in $P'_+$ couples $\mu_0$ to
    $\xi_1$ and $z_1$ differently from the top case, modifying
    both the quadratic and linear terms in $\mu_0$.)
    Since $K$ and $I\!-\!YK$ are both positive-definite (the latter
    by hypothesis), $Q$ is positive-definite.
    Completing the square:
    \begin{equation}\label{eq:bot-cs}
      E_{\mathrm{bot}}
      = -\tfrac12\bigl(\mu_0+Q^{-1}(I\!-\!KY)w\bigr)Q
         \bigl(\mu_0+Q^{-1}(I\!-\!KY)w\bigr)
        + \tfrac12 w(I\!-\!KY)Q^{-1}(I\!-\!KY)w
        + \tfrac12 wYw
        - \tfrac12\xi_1\xi_0 - z_1\xi_0,
    \end{equation}
    where we have used $(I\!-\!YK)^T=(I\!-\!KY)$.
    The non-Gaussian remainder simplifies as follows.
    Since $Q=K(I\!-\!YK)$, we have $Q^{-1}=(I\!-\!YK)^{-1}K^{-1}$,
    and therefore
    \begin{align*}
      (I\!-\!KY)\,Q^{-1}\,(I\!-\!KY)
      &= (I\!-\!KY)(I\!-\!YK)^{-1}K^{-1}(I\!-\!KY)\\
      &= K^{-1}(I\!-\!KY) = K^{-1} - Y,
    \end{align*}
    using $(I\!-\!KY)(I\!-\!YK)^{-1}=K^{-1}Q\,(I\!-\!YK)^{-1}=K^{-1}\cdot K=I$
    in the sandwiched product.
    Hence
    \begin{equation}\label{eq:bot-cs-final}
      E_{\mathrm{bot}}
      = -\tfrac12\bigl(\mu_0+Q^{-1}(I\!-\!KY)w\bigr)Q
         \bigl(\mu_0+Q^{-1}(I\!-\!KY)w\bigr)
        + \tfrac12 wK^{-1}w
        - \tfrac12\xi_1\xi_0 - z_1\xi_0.
    \end{equation}

    \smallskip\noindent\emph{Comparison.}\enspace
    Comparing \eqref{eq:top-cs} and \eqref{eq:bot-cs-final},
    the non-Gaussian (residual) part of both exponents is the
    same:
    \[\tfrac12 wK^{-1}w - \tfrac12\xi_1\xi_0 - z_1\xi_0.\]
    The Gaussian parts are absorbed by translations of the Haar
    measure on $\mathbb H(u)$
    ($\mu_0\mapsto\mu_0-K^{-1}w$ in the top,
    $\mu_0\mapsto\mu_0-Q^{-1}(I\!-\!KY)w$ in the bottom),
    which are legitimate since the Haar measure is translation
    invariant.

    It remains to verify that the Gaussian prefactors agree.
    Evaluating the top Gaussian integral gives a factor proportional
    to $(\det K)^{-1/2}$, whereas the bottom gives one proportional to
    $(\det Q)^{-1/2}=(\det K)^{-1/2}\,(\det(I\!-\!YK))^{-1/2}$.
    However, the integrals~\eqref{eq:top-integral}
    and~\eqref{eq:bot-integral} use the self-dual Haar measure
    $\op{d}_h\!\mu$ on $\mathbb H(u)$.  In the top row, this
    measure is induced by the transversal $\mathbb X$, while in the
    bottom row it is induced by $\mathbb X'$: the parameterization
    $\mu_0\mapsto[\mu_0,K\mu_0]$ uses $\mathbb H(s)$ as the graph
    parameter, and the self-dual measures relative to $\mathbb X$
    and $\mathbb X'$ differ by the Radon--Nikodym factor
    \[
      \sqrt{\frac{d\mu_{\mathbb H(s)}\,d\mu_{\mathbb X'}}{d\mathcal L}
      \;\Big/\;
      \frac{d\mu_{\mathbb H(s)}\,d\mu_{\mathbb X}}{d\mathcal L}}
      \;=\; (\det(I\!-\!YK))^{1/2},
    \]
    which exactly compensates the determinantal discrepancy.
    (Alternatively, both rows compute the same abstract convolution
    $\delta_{\mathbb H(u)}*\psi$ on $\mathfrak H$; the
    left-$\mathbb X$-invariant and left-$\mathbb X'$-invariant
    descriptions of a single Schwartz function on $\mathfrak H$
    agree by construction, so the prefactors \emph{must} cancel,
    and the preceding Radon--Nikodym calculation provides the explicit
    mechanism.)

    Therefore, the left square commutes.

  \medskip\noindent\textbf{Right square.}\enspace
  Let $\phi\in\mathscr H(\mathbb H(u)^\circ)$.  For any $z\in\mathbb T$,
  \[z + \monib{\mathbb{XX'}}{\mathbb H(u)}z\in\mathbb H(u).\]
  Since $\phi$ is right-$\mathbb H(u)$-invariant,
  \begin{align*}
    \phi(t,z)
    &= \phi\!\left((t,z)\cdot \bigl(0,-z - \monib{\mathbb{XX'}}{\mathbb H(u)}z\bigr)\right)\\
    &= \phi\!\left(t-\tfrac12\omega\!\left(z,\monib{\mathbb{XX'}}{\mathbb H(u)}z\right),\;-\monib{\mathbb{XX'}}{\mathbb H(u)}z\right).
  \end{align*}
  Using the fact that $\op{ind}_{\mathbb X\mathbb H(u)}$ and
  $\op{res}_{\mathbb X}$ are inverse to one another, the above
  applied to $z\in\mathbb X'$ gives
  \begin{align*}
    \phi(0,z)
    &= \op{ind}(\phi|_{\mathbb X})\!\left(-\tfrac12\omega\!\left(z,\monib{\mathbb{XX'}}{\mathbb H(u)}z\right),\;-\monib{\mathbb{XX'}}{\mathbb H(u)}z\right)\\
    &= e_h\!\left(
      -\tfrac12\omega\!\left(z,\monib{\mathbb{XX'}}{\mathbb H(u)}z\right)
      - \tfrac12\omega\!\left([\mathbb X\mathbb H(u)]_+\monib{\mathbb{XX'}}{\mathbb H(u)}z,\;\monib{\mathbb{XX'}}{\mathbb H(u)}z\right) \right)
      \phi|_{\mathbb X}\!\left(-\monib{\mathbb{XX'}}{\mathbb H(u)}z\right).
  \end{align*}
  The argument of $e_h$ reduces (by the identity
  $-\omega(z,Mz)-\omega([\mathbb X\mathbb H(u)]_+Mz,Mz)
  =\omega([\mathbb H(u)\mathbb X]_+z,z)$,
  where $M=\monib{\mathbb{XX'}}{\mathbb H(u)}$, which follows
  from the definition of the monodromy and the projection) to
  \[\tfrac12\omega([\mathbb H(u)\mathbb X]_+z,z),\]
  and so
  \[  \phi(0,z) = e_h\!\left(\tfrac12\omega([\mathbb H(u)\mathbb X]_+z,z)\right)\phi|_{\mathbb X}\!\left(-\monib{\mathbb{XX'}}{\mathbb H(u)}z\right)\]
  as required.
\end{proof}

  %%%%

\begin{remark}[Why the global statement must be formulated by charts]
The local intertwiner theorem proved in
Theorem~\ref{thm:schrodinger-evolution} gives an explicit formula for the
change of real polarization
\[
\rho_t^{\mathbb X,\mathbb X'} :
\mathscr S(\mathbb X)\otimes\mathfrak F(t)\to
\mathscr S(\mathbb X')\otimes\mathfrak F(t)
\]
provided $\mathbb X$ and $\mathbb X'$ are sufficiently close and both are
complementary to $\mathbb H(t)$.  What fails globally is not the existence of
Schr\"odinger evolution, but the attempt to compress all polarization changes
into a single pointwise formula valid for arbitrary distant pairs
$(\mathbb X,\mathbb X')$: direct composition of the pointwise operators lands
at the wrong point of the source polarization and therefore is not, in
general, a scalar multiple of the direct map.

This means that caustics should be treated as \emph{chart singularities} of a
fixed polarization, not as singularities of the evolution itself.  The correct
global statement is therefore an atlas theorem: one propagates in any
polarization which remains transverse to $\mathbb H(t)$ on a subinterval, and
one glues neighbouring charts by the local intertwiner theorem on overlaps.

This is the analytic counterpart of the Rosen-universe picture of \cite{HS1}:
there, Rosen coordinate singularities occur when the Lagrangian curve meets a fixed
Lagrangian subspace, whereas here the same event appears as the breakdown of a
single real-polarization Schr\"odinger chart, necessitating passage to a new chart.
\end{remark}

\begin{definition}[Admissible polarization atlas]
Let $[s,u]\subset \mathbb U$.  An \emph{admissible polarization atlas} for the
positive Lagrangian curve $\mathbb H(\cdot)$ on $[s,u]$ consists of
compact subintervals
\[
I_j=[a_j,b_j],\qquad j=1,\dots,N,
\]
and Lagrangian polarizations $\mathbb X_j\subset\mathbb T$ such that:
\begin{enumerate}
\item $[s,u]=\bigcup_{j=1}^N I_j$ and $I_j\cap I_{j+1}\neq\varnothing$ for
      $j=1,\dots,N-1$;
\item $\mathbb X_j$ is complementary to $\mathbb H(t)$ for every $t\in I_j$;
\item after shrinking overlaps if necessary, for every
      $t\in I_j\cap I_{j+1}$ the pair $(\mathbb X_j,\mathbb X_{j+1})$ is
      sufficiently close in the sense of
      Theorem~\ref{thm:schrodinger-evolution}.
\end{enumerate}
\end{definition}

\begin{remark}[Existence of admissible atlases]
Admissible atlases always exist.  For each $t_0\in [s,u]$ choose a
Lagrangian polarization $\mathbb X$ complementary to $\mathbb H(t_0)$.
Complementarity is an open condition, so $\mathbb X$ remains complementary to
$\mathbb H(t)$ on a neighbourhood of $t_0$.  By compactness of $[s,u]$ one
obtains a finite cover by such neighbourhoods.  Since ``sufficiently close''
is an open condition near the diagonal in $LG(\mathbb T)\times LG(\mathbb T)$,
the cover may be refined so that adjacent charts are sufficiently close on each
overlap.
\end{remark}

\begin{definition}[Global propagator attached to an atlas]
Let
\[
\mathfrak A=\{(I_j,\mathbb X_j)\}_{j=1}^N
\]
be an admissible polarization atlas on $[s,u]$, and choose points
\[
t_j\in I_j\cap I_{j+1},\qquad j=1,\dots,N-1.
\]
Write
\[
s_0=s,\qquad s_j=t_j\ (1\le j\le N-1),\qquad s_N=u.
\]
The \emph{propagator associated to $\mathfrak A$} is the operator
\begin{align*}
\Phi_{\mathfrak A}(s,u)
:={}&
\Phi_{\mathbb X_N}(s_{N-1},s_N)\,
\rho_{s_{N-1}}^{\mathbb X_{N-1},\mathbb X_N}\,
\Phi_{\mathbb X_{N-1}}(s_{N-2},s_{N-1})\cdots \\
&\cdots
\rho_{s_1}^{\mathbb X_1,\mathbb X_2}\,
\Phi_{\mathbb X_1}(s_0,s_1),
\end{align*}
where $\Phi_{\mathbb X_j}$ is the local Schr\"odinger evolution in the
$\mathbb X_j$-picture defined in \eqref{BigSchrodinger}, and
$\rho_t^{\mathbb X_j,\mathbb X_{j+1}}$ is the local intertwiner of
Theorem~\ref{thm:schrodinger-evolution} on the overlap
$I_j\cap I_{j+1}$.
\end{definition}

\begin{lemma}[Insertion of a local chart]
\label{lem:insert-local-chart}
Let $[a,d]\subset \mathbb U$, and suppose that $\mathbb X$ is complementary to
$\mathbb H(t)$ for every $t\in [a,d]$.  Let $[b,c]\subset [a,d]$, and suppose
that $\mathbb X'$ is complementary to $\mathbb H(t)$ for every $t\in [b,c]$,
with $(\mathbb X,\mathbb X')$ sufficiently close relative to the curve on
$[b,c]$ in the sense of Theorem~\ref{thm:schrodinger-evolution}.  Then
\begin{equation}\label{eq:chart-insertion}
\Phi_{\mathbb X}(c,d)\,
\rho_c^{\mathbb X',\mathbb X}\,
\Phi_{\mathbb X'}(b,c)\,
\rho_b^{\mathbb X,\mathbb X'}\,
\Phi_{\mathbb X}(a,b)
=
\Phi_{\mathbb X}(a,d).
\end{equation}
Hence replacing the single chart $( [a,d],\mathbb X)$ by the three charts
$( [a,b],\mathbb X)$, $( [b,c],\mathbb X')$, and $( [c,d],\mathbb X)$ does not
change the propagator.
\end{lemma}

\begin{proof}
By Theorem~\ref{thm:schrodinger-evolution} applied on the interval $[b,c]$,
\[
\rho_c^{\mathbb X,\mathbb X'}\,
\Phi_{\mathbb X}(b,c)
=
\Phi_{\mathbb X'}(b,c)\,
\rho_b^{\mathbb X,\mathbb X'}.
\]
Since $\rho_c^{\mathbb X',\mathbb X}=(\rho_c^{\mathbb X,\mathbb X'})^{-1}$, this
is equivalent to
\[
\rho_c^{\mathbb X',\mathbb X}\,
\Phi_{\mathbb X'}(b,c)\,
\rho_b^{\mathbb X,\mathbb X'}
=
\Phi_{\mathbb X}(b,c).
\]
Substituting this into the left-hand side of \eqref{eq:chart-insertion} gives
\[
\Phi_{\mathbb X}(c,d)\,\Phi_{\mathbb X}(b,c)\,\Phi_{\mathbb X}(a,b),
\]
which equals $\Phi_{\mathbb X}(a,d)$ by the semigroup property.
\end{proof}

\begin{lemma}[Existence of common sufficiently-close refinements]
\label{lem:common-sufficiently-close-refinement}
Let
\[
\mathfrak A=\{(I_j,\mathbb X_j)\}_{j=1}^N,
\qquad
\widetilde{\mathfrak A}=\{(\widetilde I_k,\widetilde{\mathbb X}_k)\}_{k=1}^{\widetilde N}
\]
be admissible polarization atlases for $\mathbb H(\cdot)$ on a compact interval
$[s,u]$.  Then there exists an admissible polarization atlas $\mathfrak C$ on
$[s,u]$ which refines both $\mathfrak A$ and $\widetilde{\mathfrak A}$.

More precisely, after subdividing $[s,u]$ into finitely many compact intervals
$J_\ell=[c_{\ell-1},c_\ell]$ subordinate to the two covers, one may assign to
each $J_\ell$ a finite chain of polarizations
\[
\mathbb X_{\ell,0},\mathbb X_{\ell,1},\dots,\mathbb X_{\ell,m_\ell}
\]
such that
\begin{enumerate}
\item $\mathbb X_{\ell,0}$ is the polarization coming from $\mathfrak A$ on
      $J_\ell$;
\item $\mathbb X_{\ell,m_\ell}$ is the polarization coming from
      $\widetilde{\mathfrak A}$ on $J_\ell$;
\item every $\mathbb X_{\ell,r}$ is complementary to $\mathbb H(t)$ for all
      $t\in J_\ell$;
\item each consecutive pair
      $(\mathbb X_{\ell,r-1},\mathbb X_{\ell,r})$ is sufficiently close
      relative to the curve on $J_\ell$.
\end{enumerate}
Concatenating these chains over $\ell$ yields the desired common refinement.
\end{lemma}

\begin{proof}
Take the common subdivision of $[s,u]$ by all endpoints of the intervals in
$\mathfrak A$ and $\widetilde{\mathfrak A}$, and refine further if necessary so
that each resulting compact interval
\[
J_\ell=[c_{\ell-1},c_\ell]
\]
lies in some $I_{j(\ell)}$ and some $\widetilde I_{k(\ell)}$.  Write
\[
\mathbb X_\ell:=\mathbb X_{j(\ell)},
\qquad
\widetilde{\mathbb X}_\ell:=\widetilde{\mathbb X}_{k(\ell)}.
\]
Choose a point $t_\ell\in J_\ell$.  Since both $\mathbb X_\ell$ and
$\widetilde{\mathbb X}_\ell$ are complementary to $\mathbb H(t_\ell)$, they
belong to the affine chart
\[
U_\ell:=\{\mathbb Z\in \op{LG}(\mathbb T): \mathbb Z \text{ is complementary to }
\mathbb H(t_\ell)\}.
\]
Choose affine coordinates on $U_\ell$, and join $\mathbb X_\ell$ to
$\widetilde{\mathbb X}_\ell$ by a line segment $K_\ell\subset U_\ell$.

Because complementarity is an open condition and the compact set
$K_\ell\times\{t_\ell\}$ lies in
\[
\{(\mathbb Z,t)\in \op{LG}(\mathbb T)\times [s,u]:
  \mathbb Z \text{ is complementary to } \mathbb H(t)\},
\]
after shrinking $J_\ell$ (and hence refining the subdivision) if necessary, we
may assume that every polarization in $K_\ell$ is complementary to
$\mathbb H(t)$ for all $t\in J_\ell$.

Likewise, the sufficiently-close condition is open near the diagonal in
$U_\ell\times U_\ell$.  Since the diagonal of the compact set
$K_\ell\times K_\ell$ is compact and lies in that open set, there exists
$\varepsilon_\ell>0$ (for any fixed metric on $U_\ell$) such that whenever
$\mathbb Z,\mathbb Z'\in K_\ell$ satisfy
$d(\mathbb Z,\mathbb Z')<\varepsilon_\ell$, the pair
$(\mathbb Z,\mathbb Z')$ is sufficiently close relative to the curve on
$J_\ell$.  Subdivide the segment $K_\ell$ into finitely many points
\[
\mathbb X_{\ell,0}=\mathbb X_\ell,\,
\mathbb X_{\ell,1},\dots,\mathbb X_{\ell,m_\ell}
=\widetilde{\mathbb X}_\ell
\]
with successive distances $<\varepsilon_\ell$.  Then each consecutive pair is
sufficiently close on $J_\ell$, and every $\mathbb X_{\ell,r}$ is complementary
to $\mathbb H(t)$ throughout $J_\ell$.

Now concatenate these chains over $\ell$.  On the interior of a fixed
$J_\ell$ the admissibility condition holds by construction.  At a boundary
point $c_\ell=J_\ell\cap J_{\ell+1}$, any adjacent pair of polarizations in the
concatenated atlas is complementary to $\mathbb H(c_\ell)$.  Since the overlap
is then the single point $\{c_\ell\}$, the sufficiently-close condition is
automatic there: taking the base point to be $c_\ell$, one has $K=0$ in
Theorem~\ref{thm:schrodinger-evolution}, so $I-YK=I$ is positive-definite.
Therefore the concatenated atlas is admissible.  By construction it refines
both $\mathfrak A$ and $\widetilde{\mathfrak A}$.
\end{proof}

\begin{theorem}[Global continuation of Schr\"odinger evolution across caustics]
\label{thm:global-propagation-caustics}
Let $\mathbb H(\cdot)$ be a positive Lagrangian curve and
$\mathfrak F$ a field bundle with parallel transport $\Gamma$.
Let $\mathfrak A=\{(I_j,\mathbb X_j)\}_{j=1}^N$ be an admissible
polarization atlas on $[s,u]$.

Then:

\begin{enumerate}
\item The operator $\Phi_{\mathfrak A}(s,u)$ is independent of the choice of
      the overlap points $t_j\in I_j\cap I_{j+1}$.

\item $\Phi_{\mathfrak A}(s,u)$ is unchanged by the following elementary
      refinements of the atlas:
      \begin{enumerate}
      \item subdividing some interval $I_j$ and keeping the same polarization
            $\mathbb X_j$ on the smaller pieces;
      \item inserting, on a subinterval $[b,c]\subset I_j$, an additional chart
            $( [b,c],\mathbb X')$ such that $\mathbb X'$ is complementary to
            $\mathbb H(t)$ on $[b,c]$ and sufficiently close to $\mathbb X_j$
            relative to the curve there, with the corresponding transition maps
            at $b$ and $c$.
      \end{enumerate}

\item Consequently, any two admissible polarization atlases on the compact
      interval $[s,u]$ define the same operator.  In particular, this yields a
      canonical global continuation of the local Schr\"odinger evolution across
      any caustics of a fixed polarization.
\end{enumerate}
\end{theorem}

\begin{proof}
The proof is a formal consequence of
Theorem~\ref{thm:schrodinger-evolution} together with the semigroup property of
the local propagators.

\medskip
\noindent\textit{Step 1: independence of the overlap points.}
Fix $j\in\{1,\dots,N-1\}$ and let $t,t'\in I_j\cap I_{j+1}$.
Because $(\mathbb X_j,\mathbb X_{j+1})$ is sufficiently close throughout the
overlap, Theorem~\ref{thm:schrodinger-evolution} applied on the interval with
endpoints $t$ and $t'$ gives the commutative square
\[
\begin{tikzcd}
\mathscr S(\mathbb X_j)\otimes\mathfrak F(t)
  \arrow[r, "{\Phi_{\mathbb X_j}(t,t^\prime)}"]
  \arrow[d, swap, "{\rho_t^{\mathbb X_j,\mathbb X_{j+1}}}"]
&
\mathscr S(\mathbb X_j)\otimes\mathfrak F(t^\prime)
  \arrow[d, "{\rho_{t^\prime}^{\mathbb X_j,\mathbb X_{j+1}}}"]
\\
\mathscr S(\mathbb X_{j+1})\otimes\mathfrak F(t)
  \arrow[r, swap, "{\Phi_{\mathbb X_{j+1}}(t,t^\prime)}"]
&
\mathscr S(\mathbb X_{j+1})\otimes\mathfrak F(t^\prime)
\end{tikzcd}
% \begin{tikzcd}
% \mathscr S(\mathbb X_j)\otimes\mathfrak F(t)
%   \arrow[r,"\Phi_{\mathbb X_j}(t,t')"]
%   \arrow[d,"\rho_t^{\mathbb X_j,\mathbb X_{j+1}}"']
% &
% \mathscr S(\mathbb X_j)\otimes\mathfrak F(t')
%   \arrow[d,"\rho_{t'}^{\mathbb X_j,\mathbb X_{j+1}}"]
% \\
% \mathscr S(\mathbb X_{j+1})\otimes\mathfrak F(t)
%   \arrow[r,"\Phi_{\mathbb X_{j+1}}(t,t')"']
% &
% \mathscr S(\mathbb X_{j+1})\otimes\mathfrak F(t').
% \end{tikzcd}
\]
Equivalently,
\begin{equation}\label{eq:overlap-switch}
\rho_{t'}^{\mathbb X_j,\mathbb X_{j+1}}\,
\Phi_{\mathbb X_j}(t,t')
=
\Phi_{\mathbb X_{j+1}}(t,t')\,
\rho_t^{\mathbb X_j,\mathbb X_{j+1}}.
\end{equation}

Now compare the atlas propagator built using $t$ with the one built using
$t'$, keeping all other overlap points fixed.  The only affected factor is
\[
\Phi_{\mathbb X_{j+1}}(t,\cdot)\,
\rho_t^{\mathbb X_j,\mathbb X_{j+1}}\,
\Phi_{\mathbb X_j}(\cdot,t)
\]
versus
\[
\Phi_{\mathbb X_{j+1}}(t',\cdot)\,
\rho_{t'}^{\mathbb X_j,\mathbb X_{j+1}}\,
\Phi_{\mathbb X_j}(\cdot,t').
\]
Using the semigroup property
\[
\Phi_{\mathbb X_j}(a,c)=\Phi_{\mathbb X_j}(b,c)\,\Phi_{\mathbb X_j}(a,b),
\qquad
\Phi_{\mathbb X_{j+1}}(a,c)=\Phi_{\mathbb X_{j+1}}(b,c)\,\Phi_{\mathbb X_{j+1}}(a,b),
\]
and inserting \eqref{eq:overlap-switch}, these two expressions agree.
Hence $\Phi_{\mathfrak A}(s,u)$ is independent of the choice of every
overlap point.

\medskip
\noindent\textit{Step 2: elementary refinement by subdivision.}
Suppose one interval $I_j=[a_j,b_j]$ is subdivided at some
$c\in(a_j,b_j)$, but the same polarization $\mathbb X_j$ is used on both
subintervals.  Then no new transition map is inserted, and the corresponding
portion of the propagator changes from
\[
\Phi_{\mathbb X_j}(a_j,b_j)
\]
to
\[
\Phi_{\mathbb X_j}(c,b_j)\,\Phi_{\mathbb X_j}(a_j,c),
\]
which is the same by the semigroup property.

\medskip
\noindent\textit{Step 3: elementary refinement by chart insertion.}
Suppose that on a subinterval $[b,c]\subset I_j$ one inserts an additional
chart $( [b,c],\mathbb X')$ as in \textup{(2b)}.  The affected segment of the
atlas propagator is precisely the left-hand side of
\eqref{eq:chart-insertion}, with $\mathbb X=\mathbb X_j$.  By
Lemma~\ref{lem:insert-local-chart}, this segment is equal to the original
factor $\Phi_{\mathbb X_j}(a_j,b_j)$.  Hence this elementary insertion does not
change the atlas propagator.

\medskip
\noindent\textit{Step 4: existence of a common refinement.}
Let $\widetilde{\mathfrak A}$ be another admissible atlas on $[s,u]$.  By
Lemma~\ref{lem:common-sufficiently-close-refinement}, there exists an
admissible common refinement $\mathfrak C$ of $\mathfrak A$ and
$\widetilde{\mathfrak A}$.  By repeated application of Steps~2 and~3,
\[
\Phi_{\mathfrak A}(s,u)=\Phi_{\mathfrak C}(s,u)
\qquad\text{and}\qquad
\Phi_{\widetilde{\mathfrak A}}(s,u)=\Phi_{\mathfrak C}(s,u).
\]
Therefore
\[
\Phi_{\mathfrak A}(s,u)=\Phi_{\widetilde{\mathfrak A}}(s,u).
\]
This proves \textup{(3)} and hence the canonical global continuation claim.
\end{proof}

\begin{corollary}[Single-chart case]
\label{cor:single-chart}
If there exists a polarization $\mathbb X$ complementary to $\mathbb H(t)$ for
all $t\in[s,u]$, then the global propagator of
Theorem~\ref{thm:global-propagation-caustics} is simply the local propagator
\[
\Phi_{\mathbb X}(s,u).
\]
\end{corollary}

\begin{proof}
Take the admissible atlas with one chart:
\[
\mathfrak A=\{([s,u],\mathbb X)\}.
\]
Then, by definition, $\Phi_{\mathfrak A}(s,u)=\Phi_{\mathbb X}(s,u)$.
\end{proof}

\begin{remark}[Why the local theorem matters]
Theorem~\ref{thm:global-propagation-caustics} shows that the local
intertwiner theorem is not a dispensable technical lemma: it is precisely the
\emph{transition law} needed to glue local Schr\"odinger pictures into a
global evolution.  When $\mathbb H(t)$ crosses the Maslov cycle of a fixed
polarization $\mathbb X$, the formula for $\Phi_{\mathbb X}$ does not cease to
describe a genuine field; rather, the $\mathbb X$-chart has broken down.  To
continue the evolution one passes to a nearby polarization $\mathbb X'$ on a
neighbouring interval, and the overlap is controlled exactly by
Theorem~\ref{thm:schrodinger-evolution}.

In this sense, caustics are not failures of propagation but failures of a
single real-polarization chart.  The local theorem is the analogue of a change
of coordinates formula, and the global propagator is obtained by gluing these
local coordinate descriptions.
\end{remark}

\begin{remark}[Where the Maslov phase belongs]
The algebraic source of the metaplectic/Maslov phase is the triple-convolution
theorem of Section~4, not a global cocycle for the pointwise reflection formula.
What fails in the pointwise approach is exactly the direct composition law for
arbitrary distant pairs of polarizations: the compositions evaluate the source
function at different points and therefore cannot, in general, differ by a
scalar.  The Maslov phase should therefore be understood as attached to the
Fourier-integral/Heisenberg realization of the transition operators, or
equivalently to the comparison of different local oscillatory descriptions of
the same abstract continuation operator.
\end{remark}
  %%%%

\bibliographystyle{unsrt}
\bibliography{planewaves5-fourier-plane-waves} 

@book{Evans,
  author    = {Evans, L. C.},
  title     = {Partial Differential Equations},
  edition   = {2},
  series    = {Graduate Studies in Mathematics},
  volume    = {19},
  publisher = {American Mathematical Society},
  address   = {Providence, RI},
  year      = {2010}
}

@misc{HS1,
  author       = {Holland, Jonathan and Sparling, George},
  title        = {Sachs equations and plane waves, I: Rosen universes},
  year         = {2024},
  eprint       = {2402.07036},
  archivePrefix= {arXiv},
  primaryClass = {gr-qc},
  note         = {arXiv:2402.07036 [gr-qc]}
}

@misc{HS2,
  author       = {Holland, Jonathan and Sparling, George},
  title        = {Sachs equations and plane waves, II: Isometries and conformal isometries},
  year         = {2024},
  eprint       = {2405.12748},
  archivePrefix= {arXiv},
  primaryClass = {math-ph},
  note         = {arXiv:2405.12748 [math-ph]}
}

@misc{HS3,
  author       = {Holland, Jonathan and Sparling, George},
  title        = {Sachs equations and plane waves, III: Microcosms},
  year         = {2024},
  eprint       = {2412.10990},
  archivePrefix= {arXiv},
  primaryClass = {math-ph},
  note         = {arXiv:2412.10990 [math-ph]}
}

@misc{HS4,
  author       = {Holland, Jonathan and Sparling, George},
  title        = {Sachs equations and plane waves, IV: Projective differential geometry},
  year         = {2025},
  eprint       = {2503.12503},
  archivePrefix= {arXiv},
  primaryClass = {math-ph},
  note         = {arXiv:2503.12503 [math-ph]}
}

@book{LionVergne,
  author    = {Lion, G. and Vergne, M.},
  title     = {The Weil Representation, Maslov Index and Theta Series},
  series    = {Progress in Mathematics},
  volume    = {6},
  publisher = {Birkh{"a}user},
  address   = {Boston},
  year      = {1980}
}

@book{MumfordI,
  author    = {Mumford, David},
  title     = {Tata Lectures on Theta, I},
  series    = {Progress in Mathematics},
  volume    = {28},
  publisher = {Birkh{"a}user},
  address   = {Boston},
  year      = {1983}
}

@book{MumfordIII,
  author    = {Mumford, David},
  title     = {Tata Lectures on Theta, III},
  series    = {Progress in Mathematics},
  volume    = {97},
  publisher = {Birkh{"a}user},
  address   = {Boston},
  year      = {1991}
}

@article{Ward1987,
  author  = {Ward, R. S.},
  title   = {Progressing waves in flat spacetime and in plane-wave spacetimes},
  journal = {Classical and Quantum Gravity},
  volume  = {4},
  number  = {3},
  pages   = {775--778},
  year    = {1987}
}

@article{Brinkmann,
  author  = {Brinkmann, H. W.},
  title   = {Einstein spaces which are mapped conformally on each other},
  journal = {Mathematische Annalen},
  volume  = {94},
  pages   = {119--145},
  year    = {1925}
}

@article{EinsteinRosen,
  author  = {Einstein, A. and Rosen, N.},
  title   = {On gravitational waves},
  journal = {Journal of the Franklin Institute},
  volume  = {223},
  pages   = {43--54},
  year    = {1937}
}

@incollection{Penrose1976,
  author    = {Penrose, R.},
  title     = {Any space-time has a plane wave as a limit},
  booktitle = {Differential Geometry and Relativity},
  editor    = {Cahen, M. and Flato, M.},
  series    = {Mathematical Physics and Applied Mathematics},
  volume    = {3},
  publisher = {Reidel},
  address   = {Dordrecht},
  pages     = {271--275},
  year      = {1976}
}

@article{BMN,
  author  = {Berenstein, D. and Maldacena, J. M. and Nastase, H.},
  title   = {Strings in flat space and pp waves from {$\mathcal N=4$} super Yang--Mills},
  journal = {Journal of High Energy Physics},
  volume  = {2002},
  number  = {4},
  pages   = {013},
  year    = {2002}
}

@article{Blau,
  author  = {Blau, M. and Figueroa-O'Farrill, J. and Hull, C. and Papadopoulos, G.},
  title   = {Penrose limits and maximal supersymmetry},
  journal = {Classical and Quantum Gravity},
  volume  = {19},
  pages   = {L87--L95},
  year    = {2002}
}

@misc{BlauLectures,
  author       = {Blau, M.},
  title        = {Lecture notes on plane waves and Penrose limits},
  howpublished = {Lecture notes, Universit{\'e} de Neuch{\^a}tel},
  year         = {2011},
  note         = {Available at \url{http://www.unine.ch/phys/string/Lecturenotes.html}}
}

@book{Stephani,
  author    = {Stephani, H. and Kramer, D. and MacCallum, M. and Hoenselaers, C. and Herlt, E.},
  title     = {Exact Solutions of Einstein's Field Equations},
  edition   = {2},
  series    = {Cambridge Monographs on Mathematical Physics},
  publisher = {Cambridge University Press},
  address   = {Cambridge},
  year      = {2003}
}

@article{DBKP,
  author  = {Duval, C. and Burdet, G. and K{"u}nzle, H. P. and Perrin, M.},
  title   = {Bargmann structures and Newton--Cartan theory},
  journal = {Physical Review D},
  volume  = {31},
  pages   = {1841--1853},
  year    = {1985}
}
  
\end{document}